\documentclass[letterpaper,twocolumn,10pt]{article}
\usepackage{usenix2019_v3}
\usepackage{tikz}
\usepackage{amsmath}
\usepackage{enumitem}
\usepackage{lipsum}
\usepackage{amsthm}
\usepackage{algorithm}
\usepackage{algorithmic}
\usepackage{caption}
\usepackage{subcaption}

\newtheorem{theorem}{Theorem}

\usepackage{xcolor}
\usepackage{xspace}
\usepackage{multirow}
\usepackage{amssymb}
\newcommand{\ournameNoSpace}{\emph{\mbox{DeepEclipse}}}
\newcommand{\ourname}{\ournameNoSpace\xspace}

\def\etal{\emph{et al}.\xspace}
\usepackage{filecontents}

\begin{document}

\date{}

\title{\Large \bf \ourname: How to Break White-Box DNN-Watermarking Schemes\\
  }

\author{
{\rm Alessandro Pegoraro}\\
Technical University of Darmstadt\\
alessandro.pegoraro@trust.tu-darmstadt.de
\and
{\rm Carlotta Segna}\\
Technical University of Darmstadt\\
carlotta.segna@trust.tu-darmstadt.de
\and
{\rm Kavita Kumari}\\
Technical University of Darmstadt\\
kavita.kumari@trust.tu-darmstadt.de
\and
{\rm Ahmad-Reza Sadeghi}\\
Technical University of Darmstadt\\
ahmad.sadeghi@trust.tu-darmstadt.de
}

\maketitle

\begin{abstract}

\noindent
Deep Learning (DL) models have become crucial in digital transformation, thus raising concerns about their intellectual property rights. Different watermarking techniques have been developed to protect Deep Neural Networks (DNNs) from IP infringement, creating a competitive field for DNN watermarking and removal methods.

\noindent
The predominant watermarking schemes use white-box techniques, which involve modifying weights by adding a unique signature to specific DNN layers. On the other hand, existing attacks on white-box watermarking usually require knowledge of the specific deployed watermarking scheme or access to the underlying data for further training and fine-tuning. 

\noindent
We propose \ourname, a novel and unified framework designed to remove white-box watermarks. We present obfuscation techniques that significantly differ from the existing white-box watermarking removal schemes. \ourname can evade watermark detection without prior knowledge of the underlying watermarking scheme, additional data, or training and fine-tuning. Our evaluation reveals that \ourname excels in breaking multiple white-box watermarking schemes, reducing watermark detection to random guessing while maintaining a similar model accuracy as the original one. Our framework showcases a promising solution to address the ongoing DNN watermark protection and removal challenges.

\end{abstract}

\section{Introduction}
\label{sec:intro}

The rising cost of computational and engineering expenditures associated with training massive Deep Neural Networks (DNN) models has reached unprecedented levels \cite{devlin2019bert, radford2019language, raffel2023exploring, openai2023gpt}. Given that well-trained DNNs are invaluable assets for AI corporations, they face a growing threat from model embezzlement and unauthorized usage \cite{tramer2016stealing, correia2018copycat, he2020towards}. Consequently, model copyright protection has become increasingly vital and studied, also by large corporations such as IBM \cite{zhang2018protecting}, Google \cite{adi2018turning}, and Microsoft \cite{chen2019deepattest, chen2020specmark}. Hence, recently, several approaches known as DNN watermarking have emerged \cite{darvish2019deepsigns, munyer2023deeptextmark, chen2019deepmarks, Uchida_2017, liu2021watermarking, wang2021riga, kirchenbauer2023watermark, feng2020watermarking, kuribayashi2021white, guan2020reversible, botta2021neunac, chen2019deepattest, li2022defending, sablayrolles2020radioactive, zhao2021structural, xie2021deepmark, chen2021you, xiaoxuan2021meets, fan2021deepip, zhang2020passport, lim2022protect, adi2018turning, zhang2018protecting, guo2018watermarking, maung2021piracy, zhu2020secure, lao2022identification, le2020adversarial, jia2021entangled, li2019prove, lao2022deepauth, li2022untargeted, zhong2020protecting, chen2019blackmarks, mehta2022aime, li2022move, charette2022cosine, namba2019robust, 9711105, wang2022nontransferable, bansal2022certified, 9954194, 9844282, atli2021waffle, li2022fedipr, shao2023fedtracker, 9989512, ong2021protecting, cong2022sslguard, 9505220}, aimed at tracking down illicit duplicates in the open domain \cite{lukas2022sok, sun2023deep}.

\noindent
\textbf{DNN Watermarking Schemes. }DNN watermarking methods can be categorized into two primary types: black-box and white-box watermarking. While the former requires only API access for model predictions, the latter demands access to the model's internal architecture. Both watermarking categories follow two essential phases: injection and verification. During training, a secret signature (i.e., the watermark) is incorporated into the target model. Subsequently, the verification is performed by extracting the signature from the model and then comparing it with the original one kept by its owner. The positive verification of the watermark's presence and authenticity shall prove the model's ownership. This verification can further prove to a third independent party, e.g., a court jury or a government agency, the reliability of the ownership's claim~\cite{sun2023deep}. 

\noindent
In this paper, our focus is on white-box watermarking schemes. Since white-box watermarks embed a signature into the model's parameters, they require access to the suspect model when verifying the ownership~\cite{darvish2019deepsigns, wang2021riga, Uchida_2017}. The watermark is, therefore, tied intricately to the model's internal structure and architecture. One distinguishing feature of the white-box watermarking is that, given the high dimensionality of neural networks, the probability of collision for two honest model owners is highly improbable since their models should have the same weight initialization and the same watermark parameters initialization to end up with the same extracted signature~\cite{darvish2019deepsigns}.
This unique property has directed significant research attention towards white-box model watermarking, highlighting its prominence in the field~\cite{lukas2022sok, sun2023deep}. 

\noindent \textbf{Attack on White-Box Watermarking. } The recent emergence of DNN watermarking schemes has given rise to watermark removal attacks, which can be classified into three subcategories~\cite{lukas2022sok}: (i) input preprocessing, (ii) model modification, and (iii) model extraction. However, as we elaborate in detail in Section \ref{sec:related_works}, these attacks have various limitations because they typically require access to (i) data for re-training, (ii) hardware capabilities, or (iii) knowledge of the embedded watermark. 

\noindent \textbf{Our Goals and Contributions. }We address the limitations of existing white-box watermark removal methods by proposing a novel unified white-box watermark obfuscation framework \ourname. Our attacks target the most studied types of white-box watermarking schemes, namely, weights-based~\cite{Uchida_2017, wang2021riga, liu2021watermarking, chen2021you, chen2019deepmarks, feng2020watermarking, kuribayashi2021white, xie2021deepmark, chen2019deepattest} and activation-based~\cite{guan2020reversible, darvish2019deepsigns, botta2021neunac, sablayrolles2020radioactive}. In a weight-based watermarking scheme, the model owner embeds a secret message into the weights of one or multiple layers during training. In contrast, in activation-based schemes, watermarks are characterized by embedding the message into the activation maps of layers for selected samples.

\noindent
We present two white-box watermark obfuscation methods, basic and advanced, associated with different adversarial settings. The basic attack assumes a setting where a passive verifier checks the credibility of the model provided by the owner. Passive verifiers receive the position of the watermarked layer(s) from the original model owner and strictly follow the protocol of the watermarking scheme to extract the message. For this setting, we introduce a novel obfuscation technique to neutralize common white-box watermarking defenses by altering the structure and weights of the model's layers through layer splitting and reshaping while preserving the model's utility. Thus, this restructuring of the model layers hinders the passive verifier from correctly extracting the watermark. In some instances (e.g., a trial), the disputing parties may be required to disclose more information about their models to the verifier. Thus, a more active verifier that has access to more information, such as the whole model, can be present for the watermark verification. As a result, this verifier can perform a set of computations on the model's parameters to effectively undo any possible obfuscation an adversary may have employed. The advanced obfuscation attack addresses this setting and is the most robust adversarial setting, where a third-party verifier typically requires access to data and parameters for watermark verification.

For the advanced obfuscation attack, we first identify possible watermarked layers using frequency analysis and then apply more substantial model modifications to make the identified layers noisy. In frequency analysis, we extract the frequency components of the weight matrix to study the patterns that help us discriminate the watermarked layers from the non-watermarked layers, as detailed in Section \ref{subsub:freq}. The rationale behind employing frequency analysis is to minimize the potential impact on the model's effectiveness. Thus, once we determine the watermarked layers, we do not have to apply advanced obfuscation techniques to all the layers. Our extensive evaluations show that even in the worst-case scenario, we reduce the accuracy by only 4\%.

\noindent
Our attacks are designed to work for DNNs that incorporate linear or convolutional layers, or a combination of both, and operate effectively without prior knowledge of the specific watermarking scheme used.

\noindent
In summary, our key contributions are as follows:
\begin{itemize}[noitemsep]
\item We propose a novel framework, \ourname, for obfuscating DNN watermarks. Our approach significantly differs from the existing white-box watermarking removal schemes since it can evade watermark detection without prior knowledge of the underlying watermarking scheme, additional data, and/or training or fine-tuning.

\item We introduce two obfuscation attacks within our framework, each tailored to different adversarial settings. The basic obfuscation involves splitting and reshaping the layers within a DNN, while the second one builds noisy layers on top of the first, specifically designed to counter active verifiers who possess access to more information than is typically required for watermark verification (cf. \S\ref{sec:design}).

\item To aid in identifying potential layers that white-box watermarking schemes may have modified, we employ a frequency-based analysis of the model weights using Discrete Fourier Transform (DFT). This analysis not only identifies potential alterations but also helps minimize the loss in utility when applying our advanced obfuscation attack (cf. \S\ref{sec:design}).

\item We extensively evaluate our approach against a diverse range of well-known white-box watermarking methods. The results demonstrate that \ourname effectively breaks multiple white-box watermarking schemes, reducing watermark detection to the level of random guessing while maintaining a similar accuracy to the original model (cf. \S\ref{sec:eval}).

\item Finally, we conduct a comprehensive evaluation of the applicability of \ourname across various model architectures and benchmark datasets. Our findings indicate that the base obfuscation algorithm has no discernible impact on the behaviour of stolen models. Additionally, our more advanced obfuscation attack, even in the worst-case scenario, only minimally impacts the accuracy of the targeted models (cf. \S\ref{sec:eval}).
\end{itemize}
\begin{figure*}[ht]
    \centering
    \includegraphics[scale=0.53]{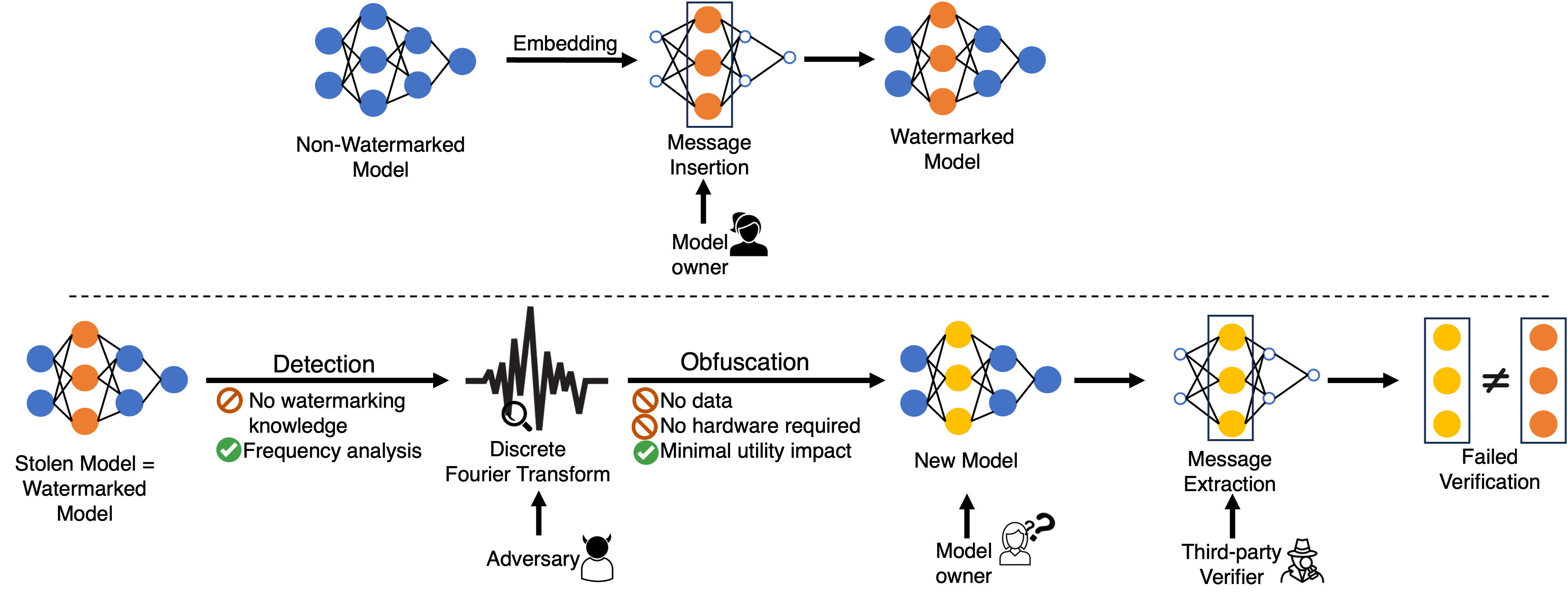}
    \caption{\ourname overview. Top: Represents the expected behavior of the model owner, i.e., watermark insertion. Bottom: Analyzed detection and proposed obfuscation pipeline, followed by the verification process. We have assumed an adversary who has acquired an unauthorized copy of the watermarked model (Stolen Model) and is trying to hinder the model's verification (done by a third-party verifier: passive or active) by executing the obfuscation techniques (base or advanced).}
    \label{fig:threat_model}
\end{figure*}

\section{Background}
\label{sec:background}

\noindent \textbf{Deep Neural Networks (DNNs). }A DNN can be expressed as a mathematical function denoted as $F(X; W)$, with $X$ representing the input data samples and $W$ denoting the network's parameters (comprising weights and biases). This network is structured into several layers, denoted as $F_i$ (where $i$ ranges from 1 to L). The first layer, known as the input layer, is labeled as $F_1$, while the final layer, termed the output layer, is designated as $F_L$. The intermediate layers are commonly referred to as hidden layers. In a feed-forward neural network, data moves in a unidirectional path, starting from the input layer, traversing through the hidden layers, and ultimately reaching the output layer. While this architecture is conventional for tasks like classification~\cite{simonyan2014very}, it is not unusual to encounter alternative architectural arrangements with different data flows~\cite{he2016deep}. Below, we briefly describe layers that have been extensively investigated in watermarking literature as the common choices for embedding watermarks.
\begin{itemize} [leftmargin=*]
    \item \textbf{Linear Layer. }A linear layer, often a fully connected or dense layer, is a fundamental element within a neural network responsible for executing a linear transformation on its input. This transformation involves taking an input vector $\overrightarrow{x} \in \mathbb{R}^{m}$, a weight matrix $A \in \mathbb{R}^{m \times n}$, and a bias vector $b \in \mathbb{R}^{n}$. The activation of the linear layer returns the output vector $\overrightarrow{x}^{'}$ with dimension $n$. In the training process, the values of $A$ and $b$ are modified/updated using backpropagation according to the patterns learned from the data. The linear transformation can be mathematically represented as:
    \begin{equation}
    \label{eq:linear_layer}
    F_{i}(\overrightarrow{x}) = \overrightarrow{x} A_{i}\;+\;b_i
    \end{equation}

    White-box watermarks can extract the message from linear layers' matrix $A$, either by checking the values of the weights itself \cite{wang2021riga}, by calculating statistical values such as its mean~\cite{chen2019deepattest}, Mean Square Error~\cite{feng2020watermarking}, Discrete Cosine Transform~\cite{kuribayashi2021white} or by evaluating the output fo the layer on specific values~\cite{darvish2019deepsigns}. Therefore, in this paper, we change the shape and values of matrix $A$ and vector $b$ to change the output and statistical values of the layer without incurring catastrophic results for the model, as detailed in Section \ref{sec:design}. 
    
    \item  \textbf{Convolutional Layer. }The convolutional layer is the base building block of convolutional neural networks (CNNs) utilized for feature extraction in image processing tasks. It performs a convolution operation on the input data to produce feature maps. Given an input signal $X$ composed of several input channels $C$ (for an input image, the channels are the 3 RGB colors), a learnable filter that is usually a square kernel $K$ of dimension $k_h \times k_w$ for each input planes, the convolution operation for a single new element at position $(\upsilon, \iota)$ of a singular feature map $p \in P$ is defined as follows:

    \begin{equation}
    \label{eq:conv_single_element}
    \begin{aligned}
    Conv_{i}(X, P, \upsilon, \iota) &= b^{P}_{i}+(X \star K^{P})(\upsilon,\iota) \\
    & =b^{P}_{i}+\sum_{\tau=0}^{C-1} \sum_{\varkappa=0}^{k_h-1} \sum_{j=0}^{k_w-1} X^{\tau}_{\upsilon + \varkappa, \iota + j}\cdot K^{P}_{\varkappa, j} 
    \end{aligned}
    \end{equation}
    
    Each element in the output feature maps undergoes iterations in the convolutional process. Convolutional layers typically consist of multiple kernels, each responsible for generating $P$ individual feature maps. These individual feature maps are then combined to form the layer's output. Thus, the convolutional layer's weights are represented as $Conv \in \mathbb{R}^{P \times C  \times k_h \times k_w}$, where $p$ depends on the number of output feature maps, and the shape $\mathbb{R}^{k_h \times k_w}$ (same for all kernels).
    Watermarking aims to insert a signature message into the kernels of specific channels within convolutional layers, as seen in the previous~\cite{Uchida_2017, chen2021you,zhao2021structural}. When extracting this signature message, it is crucial to maintain the positions of the altered weights and the statistical values of the kernel matrix. Therefore, in this work, our objective is to modify the shape and values of kernels $K$ without causing significant changes to the model, as discussed in Section \ref{sec:design}.
    
\end{itemize}

\noindent \textbf{Activation Functions. }An activation function is a mathematical operation used in various machine learning and computational models to introduce non-linearity to the output of a neuron or unit, allowing it to capture complex patterns and relationships in the data. One standard activation function of the hidden layers is the Rectified Linear Unit~\cite{agarap2018deep}, which truncates all negative values to zero. Another standard activation function that is used for the output layer is the Softmax~\cite{softmax_cit}, which transforms the output vector of $F_L$ into a probability distribution over multiple classes. \\

\noindent \textbf{Discrete Fourier Transform. }The Fourier Transform is a mathematical technique that converts a function into a representation highlighting the frequencies present in the original function. It achieves this by extracting the sine and cosine basis functions that comprise the function. The Fourier Transform can be used in both continuous (Continuous Fourier Transform, CFT) and discrete (Discrete Fourier Transform, DFT) scenarios. DFT, which includes information on both amplitude and phase. In this study, we employ the DFT~\cite{vetterli2014foundations} to break down the model's weights to analyze the pattern in the variance of changes in frequency values to distinguish the watermark from non-watermark layers, as detailed in Section \ref{sec:design}. \\

\noindent \textbf{White-box Watermarking. }A white-box watermarking scheme for DNNs is designed to transparently mark and distinguish a specific neural network model or its elements. Thus, allowing one to easily evaluate the model ownership by detecting the mark (signature message). Unlike black-box watermarks, which usually add marks to the model's output or predictions, white-box watermarks incorporate the signature message directly into the model's structure or parameters. This integration makes them more conspicuous and verifiable when interacting with the model. The white-box watermark process can be divided into multiple phases: $1)$ \textit{Integration into Model}, where the owner often adds the watermark into the model during its training process. The approaches (studied in this work) typically select convolutional layers~\cite{Uchida_2017, chen2021you,zhao2021structural} or linear layers~\cite{darvish2019deepsigns,botta2021neunac} to add the watermark. $2)$ \textit{Message extraction}, where the third-party verifier obtains a secret message from the model following the instruction of the model owner, and $3)$ \textit{Watermarking verification}, where the extracted signature is compared with the signature provided by the model owner claiming the theft. The owner's message and the extracted message ($M_o, M_e \in \{0,1\}$), each of length $N$, are compared using similarity metrics computed by adding the matching bits and normalizing the count relative to the message length.

\begin{equation}
\label{eq:verification}
\text{Sim}(M_o, M_e) = \frac{1}{N} \sum_{i=0}^{N} \begin{cases}
            1, &         \text{if } M_{o_i}=M_{e_i},\\
            0, &         \text{if } M_{o_i}\neq M_{e_i}.
    \end{cases}
\end{equation}

\noindent
A watermark is considered retained in a model if it is possible to successfully extract the same message from the model with a similarity score exceeding a certain threshold specified by the watermarking scheme. If the success rate falls below this threshold, we consider the watermark as removed. Importantly, watermarks should remain in surrogate models created based on the source model.

\noindent
In this paper, we focus on attacking these white-box watermark schemes by introducing obfuscation techniques to demonstrate how easily these defenses can be evaded. Thus effectively illustrating the urgent need to develop better watermark schemes to prevent the theft of intellectual property (IP) from DNN models.
\section{Threat Model}
\label{sec:threat}

\noindent \textbf{Adversary's Goals and Constraints. }We consider a distinct threat model concerning adversaries seeking to prevent watermark verification. Figure \ref{fig:threat_model} presents an overview that outlines the watermark obfuscation setup and the associated threat model. As illustrated in Figure \ref{fig:threat_model}, we consider an adversary who has acquired an unauthorized copy of the watermarked model and is trying to hinder the model's verification by executing the obfuscation techniques (base or advanced, as detailed in Section \ref{sec:intro}). The verification is implemented by a third-party verifier who can be classified as passive or active. Importantly, we assume that the adversary does not require (i) prior knowledge about injected watermarks, (ii) data availability for training and fine-tuning, or (iii) hardware capability for training and fine-tuning. Note that, as mentioned in Section \ref{sec:related_works}, existing attacks require access to an entire dataset for training or incur significant reductions in the model utility. 

\noindent As mentioned earlier, the adversary has acquired an unauthorized copy of a watermarked model, possessing complete knowledge of its model structure and parameters. Thus, to conceal any traces of model infringement, the adversary's objective is to obstruct the watermarking verification.
\noindent \textbf{Defense's Goal and Constraints.}To ensure a trustworthy ownership verification, a model watermark scheme should satisfy a minimum set of fundamental requirements, which are: (i) \textbf{Fidelity }ensure no degradation in neural network function due to watermarking, (ii) \textbf{Reliability }minimize false negatives in watermark detection, (iii) \textbf{Robustness }withstand model modifications without compromising watermark integrity, (iv) \textbf{Integrity }minimize false positive in watermark detection, (v) \textbf{Capacity }embed large information while meeting fidelity and reliability criteria, (vi) \textbf{Efficiency }keep overhead minimal for embedding and detection, (vii) \textbf{Security }ensure undetectable watermark presence in the network, (viii) \textbf{Generalizability }applicable in both white-box and black-box scenarios. \ourname focuses against \textbf{Reliability}, \textbf{Robustness }and \textbf{Security}, since the approach will maximize the false negative by modifying layers in which the watermark presence was detected.

\noindent
The watermarking scheme fails: (i) if the verifier, following the procedure outlined by the watermarking scheme, is unable to extract the secret message ($M_e$) from the stolen model, or (ii) the computed similarity (Equation Eq. \ref{eq:verification}) between the extracted secret message ($M_e$) and the owner's message ($M_o$) is less than a specified threshold ($\delta$), or it falls into the scenario of a random guess: $Sim(M_o,\; M_e) \leq 0.5 + \delta$.

\noindent \textbf{Verifier's Goal and Constraints. }We consider two types of verifiers: a passive verifier and an active verifier. A passive verifier receives the position of the watermarked layer(s) from the original model owner and strictly follows the protocol of the watermarking scheme to extract the message. When encountering a verification failure or not correctly extracting the watermark signature, the passive verifier halts the verification process. In this case, the attacker's goal is achieved, and the watermarking is broken. \\
An active verifier can access more information beyond what is needed for standard verification protocol (passive verifier). In this setting, the verifier can inspect the layers' parameters, trying to undo any possible obfuscation the adversary may have deployed. We analyze how \ourname's advanced obfuscation technique (used on top of the base obfuscation) can effectively prevent any detection while causing a minimal drop in the model utility (Section \ref{sec:design}).

\section{Design}
\label{sec:design}
This section explains \ourname design, which provides a detailed execution overview, highlighting the insights that enhance its effectiveness in bypassing an active verifier.
\subsection{High-Level Idea}
\label{susbec:hli}
We propose two obfuscation techniques: (1) The Base obfuscation attack and (2) The Advanced obfuscation attack. Since we focus on the most prevalent architectures to implement our attack scheme (discussed in Section \ref{sec:intro}), we concentrate on the \textit{Linear} and \textit{Convolutional} layers, which are the primary targets for white-box watermarking schemes to embed a watermark. Thus, this approach is applicable to the classes of watermarking where the defense embeds the watermarks into layers involving matrix multiplication, and the signature of the watermark is present on referential transparent layers.

\noindent
\textbf{Linear Layers. }The Base obfuscation attack alters the structure and weights ($A \in \mathbb{R}^{m \times n}$) of a model's layers ($F_{i}$) through layer splitting (using the Identity matrix, $I_{n \times n}$, as shown in Figure~\ref{fig:base_linear}), thus expanding the architecture, while preserving the model's utility. Consequently, we construct two new layers with weights and shapes incompatible with the original watermarked layer while maintaining the same behavior of the stolen model's architecture. We still consider the possibility of an active verifier checking whether the output of the aggregated layers is the same as the original watermarked one~\cite{darvish2019deepsigns, chen2019deepattest}. To falsify this verification, we apply further modifications. Specifically, we split the watermarked layer into two new layers as before and combined the second one with the following layer (i.e., the first layer after the watermarked one in the original architecture), as depicted in Figure~\ref{fig:adv_linear}. This way, we implement both base and advanced obfuscation attacks for linear layers. 

\noindent \textbf{Convolutional Layers. }In the Base obfuscation attack for the \textit{Convolutional} layer, our goal is to manipulate the filters' shape ($k_h \times k_w$) of the \textit{Convolutional} layer ($Conv_{i}$) and its values without modifying its outputs. Hence evading detection from a passive verifier. In the Advanced Obfuscation Attack, we first identify the possible watermarked layers using frequency analysis of the model weights to limit the potential drop in the model's utility. To accomplish this, we employ the Discrete Fourier Transform (Section \ref{sec:background}) to extract the frequency components present in the model weights. Our analysis reveals that the average change in the frequency values of the watermarked layers is less volatile (or less fluctuations in the values) than the weights of the non-watermarked layers. Section \ref{subsec:detail_design} and \ref{subsub:freq} provides a more detailed explanation. After identifying the potential watermarked layers, we aim to apply additional perturbation to the kernel's frame after implementing the base scheme. Thus, we first draw a random noise from normal distributions (details in Section \ref{subsec:detail_design}) and incorporate it into the values of the kernel's outer border, as portrayed in Figure~\ref{fig:adv_conv}.

\noindent
In summary, to obfuscate the watermarked layers, we perform splitting and reshaping for the \textit{Linear} layers and perform splitting and noising for the \textit{Convolutional} layers.

\subsection{Detailed Design}
\label{subsec:detail_design}
To simplify the analysis, we first provide details of the detection phase and then describe the operation of the two obfuscation techniques for both the linear and the convolutional layers.

\noindent \textbf{Detecting Watermarked Layers. }As detailed in Section \ref{sec:background}, we use DFT of the weight matrix ($A$) to extract the frequency components consisting $R_{r}^{r}$ and $R_{r}^{i}$, where $R_{r}^{r}$ denotes the real part and $R_{r}^{i}$ denotes the imaginary part of the frequency amplitude. Our goal is to examine these frequency components' average rate of change, enabling us to differentiate between watermark and non-watermark layers. Upon analysis of either $R_{r}^{r}$ or $R_{r}^{i}$, we observed that the rate of change in the frequency values for watermarked layers is notably different from non-watermarked ones, as shown in Section \ref{sec:eval}.

\noindent
Thus, we employ Simple Moving Average (SMA), a well-established technique for identifying trends or patterns in data. SMA calculates the average value of a data series within a specified moving window of data points and shifts this window across the data to create a new set of averaged values. Our approach begins by considering the input frequency vector, either $R_{r}^{r}$ or $R_{r}^{i}$, of length $f$. Then, we apply the SMA formula to assess the average rate of change in these frequency components within a window of size $r$, as described below:  

\begin{equation}
    SMA(l) = \frac{R_{r}^{i}(l) + R_{r}^{i}(l-1) + \ldots + R_{r}^{i}(l-r+1)}{r} 
\end{equation}

\noindent
where $R_{r}^{i}(l)$ represents it's value at index $l$. \\

\noindent
Next, we analyze the core assumptions or working of the watermarking schemes (discussed in Section \ref{sec:intro}) regarding the embedding and extracting of secret messages. The motivation here is to decide on different attack strategies to help remove the watermark from the models. We do this analysis for both the \textit{Linear} and \textit{Convolutional} layers.\\

\begin{figure}[!t]
    \centering
    \includegraphics[scale=0.43]{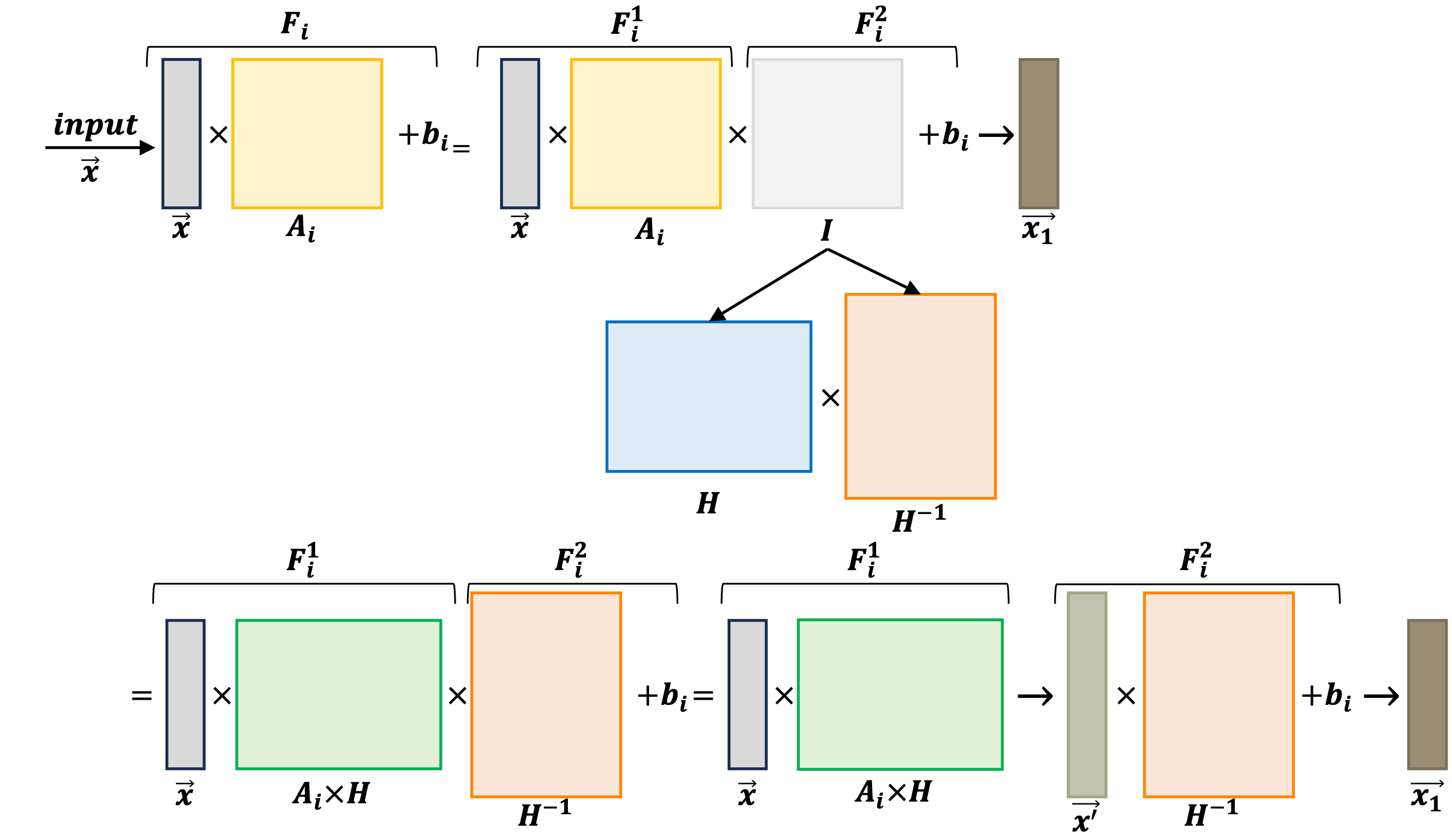}
    \caption{Basic \textit{Linear} Layers obfuscation. The original layer is split into two, with the first layer being the matrix multiplication between the original layer and a random matrix and the second layer being the inverse of the random matrix.}
    \label{fig:base_linear}
\end{figure}

\begin{figure*}[!t]
    \centering
    \includegraphics[scale=0.63]{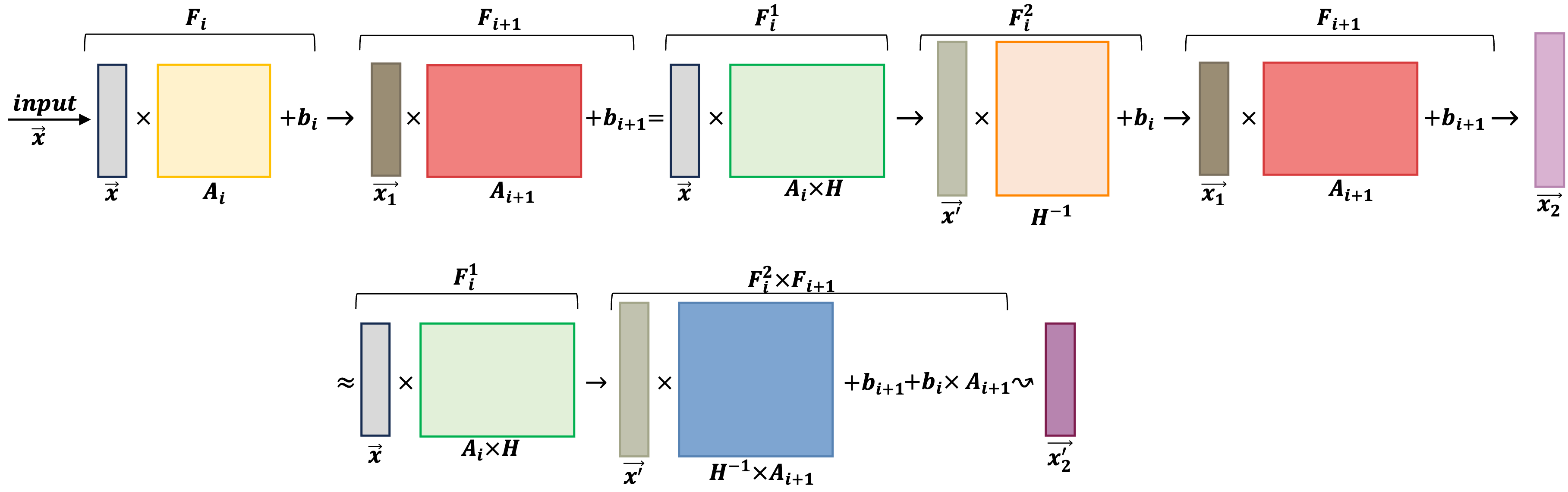}
    \caption{Advanced \textit{Linear} Layers obfuscation. The watermarked layer is multiplied by a random matrix, and the subsequent layer is multiplied by the inverse of the random matrix. The subsequent Bias is updated with the original watermarked Bias.}
    \label{fig:adv_linear}
\end{figure*}

\noindent \textbf{Linear layers. }A generic watermark defense framework needs to keep track of the original layer's weight values and their position in the $\mathbb{R}^{m \times n}$ matrix~\cite{liu2021watermarking, chen2019deepmarks, feng2020watermarking, kuribayashi2021white, xie2021deepmark, botta2021neunac} during the embedding and extracting of a message in one or multiple \textit{Linear} layers. Hence, some defenses work to address this limitation by taking an alternative approach. Such as, they employ one of two strategies: either train an extra deep neural network (DNN) to extract the hidden message from the layer's weights~\cite{sablayrolles2020radioactive} or analyze the activation maps of the layer for specific input instances~\cite{darvish2019deepsigns, chen2019deepattest}. Hence, it becomes apparent that the attack on these defenses can be divided into three categories: (i) attacks targeting the model weights, (ii) attacks targeting the DNN-based approaches, and (iii) attacks targeting the activation maps. Consequently, the principle method to remove the watermark can also be divided into three categories: first, modify the weight's matrix; second, impede the DNN-based approaches to change the shape of the weight's matrix (such that the DNN is no longer able to take them as input), and lastly, replace the output shape of the layer. Below, we propose a scheme encompassing all the above categories to fool the active verifier. \\

\noindent\textbf{Base Linear Layer Obfuscation. }Given the input vector $\overrightarrow{x}$ to the watermarked layer $F_{i}$ with weights matrix $A_{i}$, we conduct the following transformation and execute the obfuscation techniques on it to evade watermark verification by the passive verifier. Figure \ref{fig:base_linear} illustrates the Base obfuscation attack for the \textit{Linear} layer $F_{i}$. Our motivation for executing this obfuscation is to introduce an Identity matrix ($I_{n \times n} = HH^{-1}$) to transform the characteristic Equation (Equation \ref{eq:linear_layer}) of \textit{Linear} layer $F_{i}$ into two new layers: $F^{1}_{i}$ and $F^{2}_{i}$, as shown in Figure~\ref{fig:base_linear}. For a passive verifier, both will appear as random layers with no relation to the original $F_{i}$. Therefore, any signature $M_e$ extracted from either of them will not be similar or comparable with the owner secret message $M_o$. The new layers are given by:
\begin{equation}
\label{eq:identity_split}
\begin{aligned}
F^{1}_{i}(\overrightarrow{x}) &=\; \overrightarrow{x} \times A_{i}\\
F^{2}_{i}(F^{1}_{i}) &=\; F^{1}_{i}(\overrightarrow{x}) \times I_{n \times n} +\;b_i
\end{aligned}
\end{equation}

\noindent
Notably, this does not change the behaviour of the whole model, as the aggregated behaviour of the two layers is identical to the original one, $F_{i}(\overrightarrow{x}) = F^{2}_{i}(F^{1}_{i}(\overrightarrow{x}))$. For example, the original vector $\overrightarrow{x}$ is first computed with layer $F^{1}_{i}$. Then, the result is given as input to layer $F^{2}_{i}$, as shown in Equation~\ref{eq:linear_layer}. Due to the introduction of $H$ and $H^{-1}$, which can be of any size and content, $F^{1}_{i}$ and $F^{2}_{i}$ have custom shapes and parameters that appear random to a third-party passive verifier.

\noindent
However, finding an $H$ of any shape such that $HH^{-1} = I_{n \times n}$ can appear challenging, as it does not consider the case for a non-square (or rectangular) matrix $A_{i}$. The reason is that it is not prevalent in the literature that the multiplication of a rectangular matrix with its inverse is equal to an identity matrix (rectangularity problem).

\noindent
Hence, to tackle the problem of changing the values and shape of the layer's weights while accounting for the problem of rectangularity, we proceed as follows: first, we construct two rectangular matrices $H$ and its inverse $H^{-1}$ such that: (i) $H \times H^{-1} = I_{n \times n}$ and (ii) $A_{i}$ is dimensionally compatible with $H$ (i.e., we can multiply $A_{i}$ with $H$). Then, we can assign the product $A_{i} \times H$ to $F^{1}_{i}$ and $H^{-1}$ to $F^{2}_{i}$. Consequently, as per the modification depicted in Figure \ref{fig:base_linear}, from $F_{i}$ we obtain two new layers $F^{1}_{i}$ and $F^{2}_{i}$, with weights and shapes incompatible with the original $F_{i}$. These two layers maintain the same behaviour of $F_{i}$ when substituted in sequence to the architecture of the stolen model.

\noindent
Now, the problem reduces to proving that $H$ and $H^{-1}$ exist, such that the size of $H$ becomes irrelevant, and we can obtain infinitely many different such matrices with the desired characteristics. Below, we provide a formal proof that multiplying a rectangular matrix with its inverse can result in an identity matrix.

\begin{theorem}
\label{thm:rectangularity}
$H \times H^{-1} = I_{n \times n}$ and $A_{i}$ is dimensionally compatible with $H \in \mathbb{R}^{n \times h}$, when $h > n$ and $rank(H) = n$. 
\end{theorem}

\begin{proof}
Given an arbitrary rectangular matrix $H \in \mathbb{R}^{n \times h}$ with $h > n$ and it's transpose $H^{T}$, such that $HH^T = I_{n \times n}$. We assume that the matrix $H$ has the nullity zero property~\cite{meyer2023matrix}. This property of a matrix is associated with matrices that have "full rank" or are "invertible". The "full rank" property of a matrix is a characteristic that indicates that the matrix has linearly independent columns ($rank(H) = n$), and "invertible" means that a square matrix has an inverse matrix. Under these conditions, the square matrix $HH^T$, of shape $n \times n$, is invertible, such that,

\begin{equation}
\label{eq:HHT_invertible}
(HH^T) \times (HH^T)^{-1} = I_{n \times n}
\end{equation}

\noindent
Then, due to the associative property of matrix multiplication, we can modify Equation \ref{eq:HHT_invertible} and replace it by:

\begin{equation}
\label{eq:almost_H_inv}
H[H^T \times (HH^T)^{-1}] = I_{n \times n}
\end{equation}

\noindent
Thus, $H^T \times (HH^T)^{-1}$ can be replaced by $H^{-1}$. Consequently, we have found two matrices $H \in \mathbb{R}^{n \times h}$ and $H^{-1} \in \mathbb{R}^{h \times n}$ such that \mbox{$H \times H^{-1} = I_{n \times n}$} with the only constrain of (i) choosing $h$ such that \mbox{$h > n$} and (ii) the square matrix $HH^T$ must be invertible.
\end{proof}

\noindent
In summary, the constraints for finding $H$ and $H^{-1}$ becomes: (i) $h > n$ and (ii) $rank(H) = n$. It allows us to generate infinitely many random matrices, and, as shown in Figure \ref{fig:base_linear}, finally updating Equation \ref{eq:identity_split} into:

\begin{equation}
\label{eq:base_obf}
\begin{aligned}
F^{1}_{i}(\overrightarrow{x}) &=\; \overrightarrow{x} \times (A_{i} \times H)\\
F^{2}_{i}(F^{1}_{i}) &=\; F^{1}_{i}(\overrightarrow{x}) \times H^{-1} +\;b_i
\end{aligned}
\end{equation}

\noindent
Finally, we can substitute $F_{i}(\overrightarrow{x})$ with the sequence of the transformed two new layers, $F^{1}_{i}$ and $F^{2}_{i}$. $F^{1}_{i}$ can either have no bias or a bias with an imperceptible value because it adds no perturbations to the output. Further, if we choose a matrix $H$ such that $A_{i} \times H$ has only positive values, and $F_{i-1}$ uses an activation function such as $\text{ReLU}$, we can further make $F^{1}_{i}$ inconspicuous to a passive verifier by applying $\text{ReLU}(F^{1}_{i}(\overrightarrow{x}))$ after $F^{1}_{i}$. However, an active verifier (see Section \ref{sec:threat}) may undo the obfuscation if $H^{-1}$ is multiplied with $F^{2}_{i}$, as showcased in Figure \ref{fig:base_linear}. Hence, we propose an advanced layer obfuscation technique for the linear layers to evade the active verifier. We apply further modifications to the base method by keeping the first constructed layer $F^1_{i}$ and aggregating the second one $F^2_{i}$ with the subsequent \textit{Linear} layer $F_{i+1}$ if it exists. Otherwise, we can easily generate a new one using $H$ and $H^{-1}$. This way, the perceived outputs will only be the one from $F^{1}_{i}$ and $F^{2}_{i} \times F_{i+1}$, as shown in Figure~\ref{fig:adv_linear}. 

\noindent \textbf{Advanced Linear Layer Obfuscation. }To deceive the active watermark verifier that examines the output of the \textit{Linear} layers and is more robust in other adversarial scenarios, where it can employ complex computations (Section \ref{sec:threat}), which includes using multiplication to reverse the splitting of $F_{i}$ into $F^{1}_{i}, \; F^{2}_{i}$ presented above. To deceive such an active verifier, it is imperative to devise an advanced obfuscation method. To do so, as shown in Figure \ref{fig:adv_linear}, \ourname aims to conceal $F^{2}_{i}$ within the \textit{Linear} layer following $F_{i}$, denoted as $F_{i+1}$. Thus, we keep the first layer from the base obfuscation: $F^{1}_{i}$, and we substitute $F_{i+1}$ with $F_{i+1} \times F^{2}_{i}$, as shown in Equation~\ref{eq:advanced_linear}. As we have shown before for the passive verifier, the active verifier extracts a signature $M_e$ from $F_{i+1}$ that is incompatible with the $M_o$ (that the owner extracts from $A_{i}$), further, even the signature extracted from $F_{i+1} \times F^{2}_{i}$ is of no use to the active verifier. Based on the structural properties of DNNs, it is guaranteed that $A_{i}$ must be dimensionally compatible with $A_{i+1}$. We also update the bias $b_{i+1}$ by appending the bias $b_{i}$ from the now hidden layer $F^{2}_{i}$. If the original layer $F_{i}$ made use of any activation function such as $\text{ReLU}$, then the new layer $F^{'}_{i+1}$ does not maintain the soundness that we guarantee in our base obfuscation approach, because the intermediate activation function is skipped. Nonetheless, as we will show in Section \ref{sec:eval}, our obfuscation incurs minimal utility loss.
\begin{equation}
\label{eq:advanced_linear}
\begin{aligned}
F^{'}_{i+1}(\overrightarrow{x}) &=\; \overrightarrow{x} \times A_{i+1} +\;b_{i+1} \\
F^{'}_{i+1}(\overrightarrow{x}) &=\; \overrightarrow{x} \times (H^{-1} \times A_{i+1}) +\;(b_{i} \times A_{i+1}) +\;b_{i+1}
\end{aligned}
\end{equation}

\begin{figure*}[h!t]
    \centering
    \includegraphics[scale=0.7]{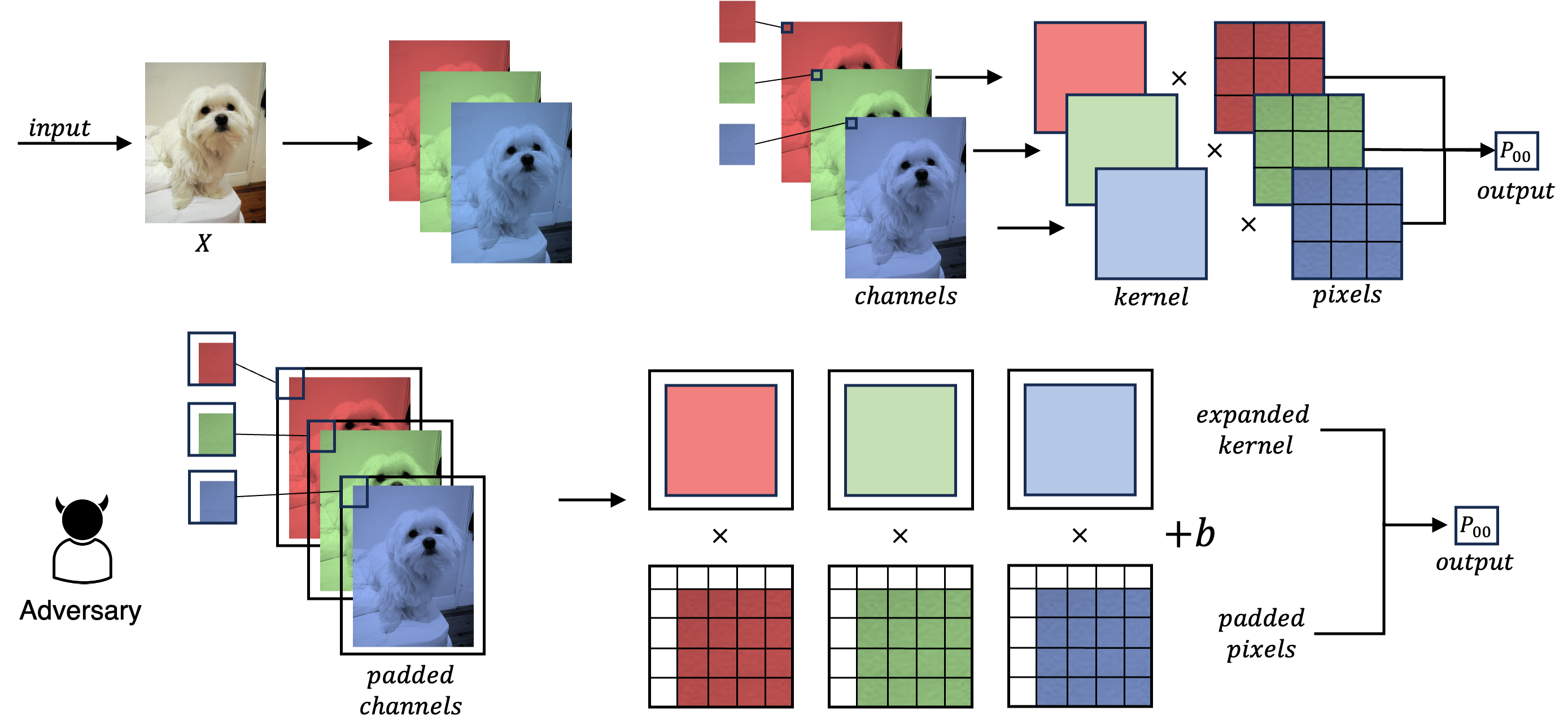}
    \caption{Basic \textit{Convolutional} Layers obfuscation. Each feature maps of the Kernel is expanded with zeros padding.}
    \label{fig:base_conv}
\end{figure*}

\begin{figure}[!t]
    \centering
    \includegraphics[scale=0.55]{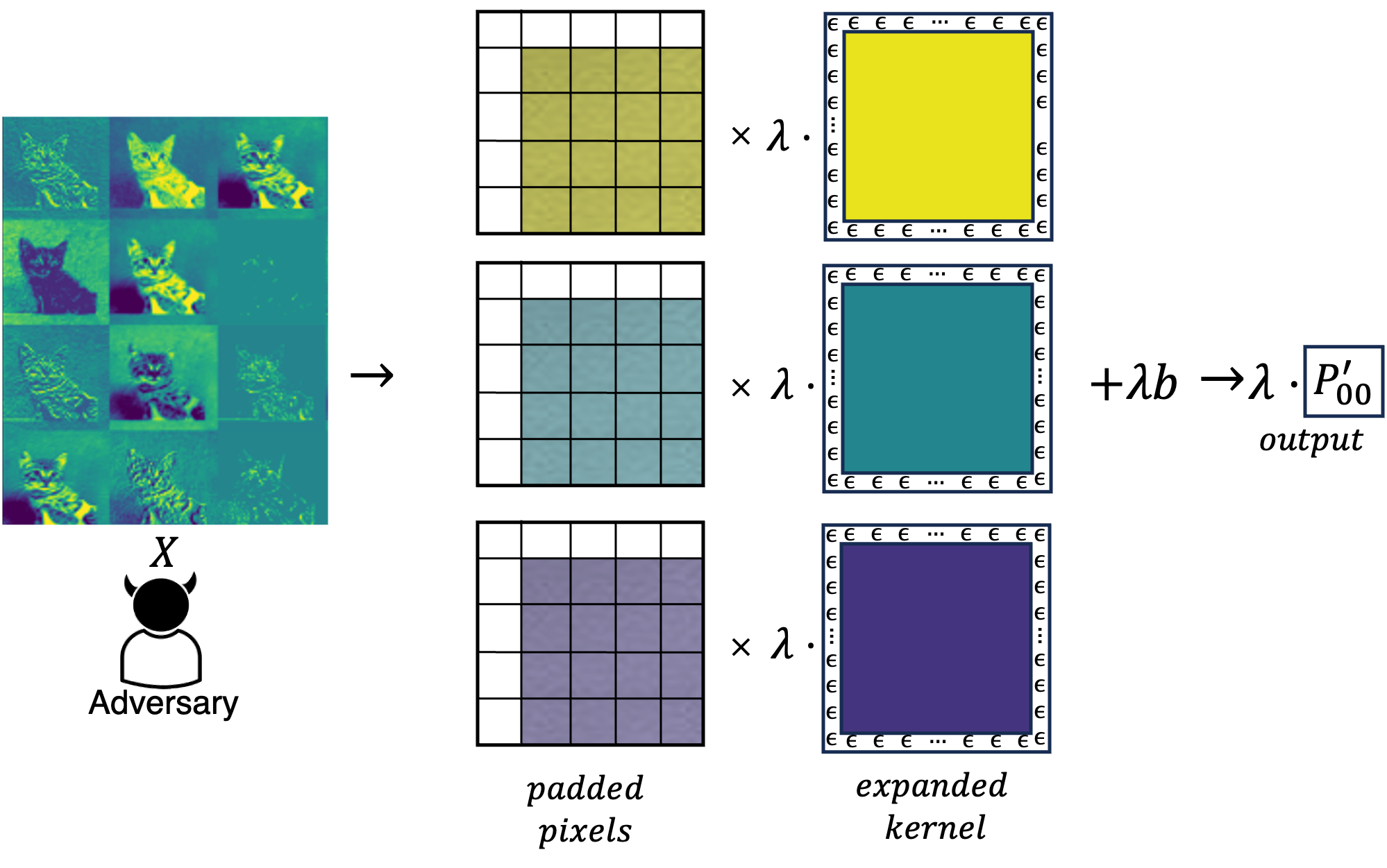}
    \caption{Advanced \textit{Convolutional} Layers obfuscation. Each feature maps of the Kernel is expanded with padding using an $\epsilon$ value, then the whole layer is multiplied by a random constant $\lambda$, and the subsequent layer is also multiplied by $\frac{1}{\lambda}$.}
    \label{fig:adv_conv}
\end{figure}

\noindent \textbf{Convolutional Layers. }After examining the defenses for DNN watermarking, it becomes apparent that a generic watermark defense framework aiming to embed and extract a message within one or multiple convolutional layers should consider various factors. These factors include the number of channels and the shape of the kernel ($\mathbb{R}^{C \times k_h \times k_w}$)~\cite{Uchida_2017, liu2021watermarking,xie2021deepmark, chen2021you, botta2021neunac, chen2019deepmarks}, the location of the selected weights within the filters~\cite{feng2020watermarking}, and the statistical values extracted from the filter's weights~\cite{xie2021deepmark, chen2021you, botta2021neunac}. Specific approaches address these vulnerabilities differently. Instead of directly encoding the message into the weights, they opt for either training an additional DNN to extract the hidden message from the layer's weights~\cite{wang2021riga} or evaluate the values of the output feature maps $p$ (Section \ref{sec:background}) for particular input signals~\cite{chen2019deepattest, guan2020reversible}. 
\noindent Hence, the attack on these defenses can be divided into three categories: (i) attacks targeting the kernel parameters (shape, location, and values), (ii) attacks targeting the DNN-based approaches, and (iii) attacks targeting the output feature maps. Consequently, the principle method to remove the watermark can also be divided into three categories: first, modify the kernel; second, impede the DNN-based approaches; and lastly, replace the output feature map. Below, we propose a scheme encompassing all the above categories to fool the active verifier. 

\noindent \textbf{Base Convolutional Layer Obfuscation. }In this passage, we target watermarking methods that embed their signature within the \textit{Convolutional} layers denoted as $Conv_{i}$. We formally prove how to manipulate the filters' shape of the \textit{Convolutional} layer and values without modifying its outputs, evading detection from passive verifiers.

\begin{theorem}
\label{thm:base_conv}
\textit{Convolutional} layer, $Conv_{i}$, filters' shape \mbox{($k_h \times k_w$)}, and values (at location $(\upsilon, \iota)$) can be manipulated without modifying their outputs.
\end{theorem}
\begin{proof}
In this proof, we aim to extend the kernel filter (Equation \ref{eq:conv_single_element}). As illustrated in Figure \ref{fig:base_conv}, one can observe that this extension is achieved by padding the initial kernel ($Conv_{i}$) with a frame of zeros (or values that vanish imperceptibly) with no impact on the layer's output. Although the frame's dimensions can vary, we will assume a frame of $-1$ on all sides for this design explanation. Still, any value, even rectangular frames, is acceptable, as non-regular shapes further help hide the obfuscation from an active verifier. 

\noindent
Then, given input matrix $X$, input channels $C$, output feature map $P$, we ensure that the behaviour of the obfuscated $Conv^{'}_{i}$ remains the same for all output feature map $P$ elements. Hence, we also consider additional padding for the input channels $C$, which must follow the same frame as the one put on the kernel (Figure \ref{fig:base_conv}). After applying these modifications, Equation \ref{eq:conv_single_element} for a single element at position ($\upsilon$, $\iota$) becomes:    

\begin{equation}
\begin{split}
\label{eq:conv_expanded}
Conv^{'}_{i}(X, P, \upsilon, \iota) = b^{P}_{i}+\sum_{\tau=0}^{C-1} \sum_{\varkappa=-1}^{k_h+1} \sum_{j=-1}^{k_w+1} X^{\tau}_{\upsilon + \varkappa, \iota + j}\cdot K^{P}_{\varkappa, j}
\end{split}
\end{equation}
\begin{multline*}
= b^{P}_{i}+\sum_{\tau=0}^{C-1} X^{\tau}_{\upsilon -1, \iota -1}\cdot K^{P}_{-1, -1}+...+X^{\tau}_{\upsilon, \iota}\cdot K^{P}_{0, 0}+...\\
...+X^{\tau}_{\upsilon + k_h -1, \iota + k_w -1}\cdot K^{P}_{k_h-1, k_w-1}+...+X^{\tau}_{\upsilon + k_h, \iota + k_w}\cdot K^{P}_{k_h, k_w}
\end{multline*}
\begin{multline*}
= b^{P}_{i}+\sum_{\tau=0}^{C-1} X^{\tau}_{\upsilon -1, \iota -1}\cdot 0+...+X^{\tau}_{\upsilon, \iota}\cdot K^{P}_{0, 0}+...\\
...+X^{\tau}_{\upsilon + k_h -1, \iota + k_w -1}\cdot K^{P}_{k_h-1, k_w-1}+...+X^{\tau}_{\upsilon + k_h, \iota + k_w}\cdot 0
\end{multline*}
\begin{equation*}
= b^{P}_{i}+\sum_{\tau=0}^{C-1} X^{\tau}_{\upsilon, \iota}\cdot K^{P}_{0, 0}+...+X^{\tau}_{\upsilon 
+ k_h -1, \iota + k_w -1}\cdot K^{P}_{k_h-1, k_w-1}
\end{equation*}
\begin{equation*}
= b^{P}_{i}+\sum_{\tau=0}^{C-1} \sum_{\varkappa=0}^{k_h} \sum_{j=0}^{k_w} X^{\tau}_{\upsilon + \varkappa, \iota + j}\cdot K^{P}_{\varkappa, j} = Conv_{i}(X, P, \upsilon, \iota)   
\end{equation*}

\noindent Thus, proving that it is possible to change the filters' shape and values of a \textit{Convolutional} layer without modifying its outputs.
\end{proof}

\noindent
However, an active verifier could notice the padding added to the filters, choosing to apply the verification on the non-padded kernel, undoing the obfuscation. Hence, below, we provide advanced layer obfuscation techniques for the \textit{Convolutional} layers to fool the active verifier. \\

\noindent \textbf{Advanced Convolutional Layer Obfuscation. }While introducing padding to the input channels of a \textit{Convolutional} layer would not typically raise suspicion of passive verifiers, as it is common practice in model architectures~\cite{he2016deep}. However, an active verifier may detect the presence of a zero border (or imperceptibly vanishing values) within the kernel. To counteract this potential detection, we develop a new formulation that introduces an additional perturbation to the kernel's border as depicted in Figure~\ref{fig:adv_conv}. The idea is to incorporate random noise into the values of the kernel's border, as shown in Figure \ref{fig:adv_conv}. As a result, we dynamically calculate distinct random noise, denoted as $\varepsilon$, for each filter $K^p \; \forall p \in P$, and add it to the kernel's border. 

\begin{equation}
\label{eq:conv_epsilon}
\begin{aligned}
&\forall p \in P,\; \varepsilon^p = \beta * median(|K^p|) * \mathcal{N}(\mu,\,\sigma^{2})\\
&\forall p \in P,\; \varepsilon^p = \beta * min(|K^p|) * \mathcal{N}(\mu,\,\sigma^{2}) 
\end{aligned}
\end{equation}

\noindent
We present two potential formulations for computing $\varepsilon$, one using the absolute median and the other using the absolute minimum. When the absolute median is chosen, it results in a greater decrease in utility compared to using the absolute minimum. Consequently, employing the absolute median leads to more pronounced changes in the statistical characteristics of the \textit{Convolutional} layer than when using the absolute minimum. In our evaluation, we employed $\mu = 0.33$, $\sigma = 0.1$, and fixed scaling of $\beta = 10$. Furthermore, to blend and integrate the values of $\varepsilon$ with the original filters $K$, we introduce an additional random factor $\lambda$ that multiplies the matrices of the new kernel as well as the bias. This adjustment elevates the values of the border, exceeding the extremely subtle vanishing values that an active verifier might erroneously perceive as zeros. Following these supplementary obfuscation steps, Equation \ref{eq:conv_expanded} transforms into:

\begin{equation}
\label{eq:conv_lambda}
\begin{aligned}
Conv^{'}_{i}(X, P, \upsilon, \iota) &= \lambda b^{P}_{i}+\sum_{\tau=0}^{C-1} \sum_{\varkappa=-1}^{k_h+1} \sum_{j=-1}^{k_w+1} X^{\tau}_{\upsilon + \varkappa, \iota + j}\cdot \lambda K^{P}_{\varkappa, j}\\ 
& = \lambda b^{P}_{i}+\lambda \sum_{\tau=0}^{C-1} \sum_{\varkappa=-1}^{k_h+1} \sum_{j=-1}^{k_w+1} X^{\tau}_{\upsilon + \varkappa, \iota + j}\cdot K^{P}_{\varkappa, j}
\end{aligned}
\end{equation}

\noindent
To prevent model fluctuations that can occur due to the random boosting of $\lambda$ in $Conv^{'}_{i}$, we divide the subsequent $Conv_{i+1}$ kernel's matrices by $\lambda$ again. Thus, following the flow of the new obfuscated model, $Conv^{'}_{i}$ and $Conv_{i+1}$ becomes:

\begin{equation}
\label{eq:conv_varepsilon}
\begin{aligned}
Conv^{'}_{i}(X, P, \upsilon, \iota) = \lambda b^{P}_{i}+\lambda \sum_{\tau=0}^{C-1} \sum_{\varkappa=-1}^{k_h+1} \sum_{j=-1}^{k_w+1} X^{\tau}_{\upsilon + \varkappa, \iota + j}\cdot K^{P}_{\varkappa, j}\\
Conv_{i+1}(X, P, \upsilon, \iota) = b^{P}_{i+1}+\frac{1}{\lambda} \sum_{\tau=0}^{C-1} \sum_{\varkappa=0}^{k_h} \sum_{j=0}^{k_w} X^{\tau}_{\upsilon + \varkappa, \iota + j}\cdot K^{P}_{\varkappa, j}
\end{aligned}
\end{equation}

\noindent
If the subsequent layer is \textit{Linear}, we can employ the base obfuscation technique as shown above to obtain $HH^{-1} = \frac{1}{\lambda} \cdot I$. As the random noise $\varepsilon$ is added to the layers, we cannot guarantee the soundness of this advanced obfuscation approach. Nonetheless, as shown in Section \ref{sec:eval}, while the obfuscation incurs some utility loss, the verification of all watermarking is compromised.

\begin{figure}[t]
    \centering
    \begin{subfigure}[b]{0.2\textwidth}
         \centering
         \includegraphics[width=\textwidth, trim=0.2cm 0.4cm 0.2cm 0]{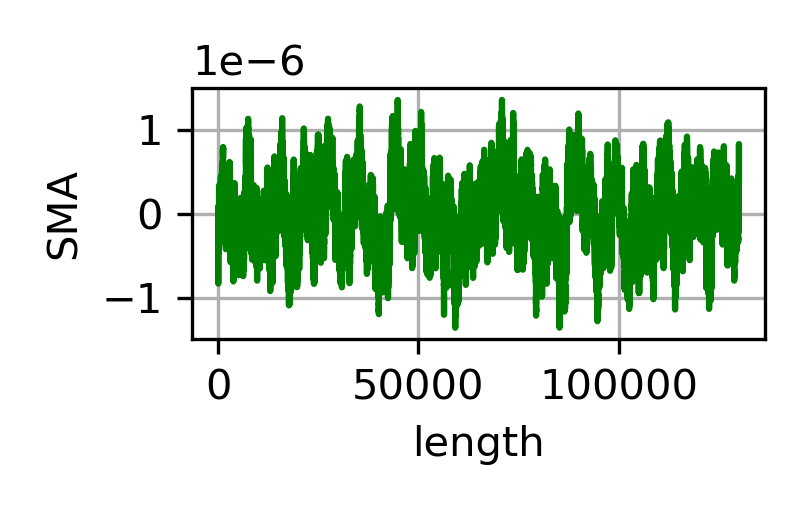}
         \caption{}
  	\label{fig:sma5_17}
     \end{subfigure}
    \begin{subfigure}[b]{0.2\textwidth}
         \centering
         \includegraphics[width=\textwidth, trim=0.2cm 0.4cm 0.2cm 0]{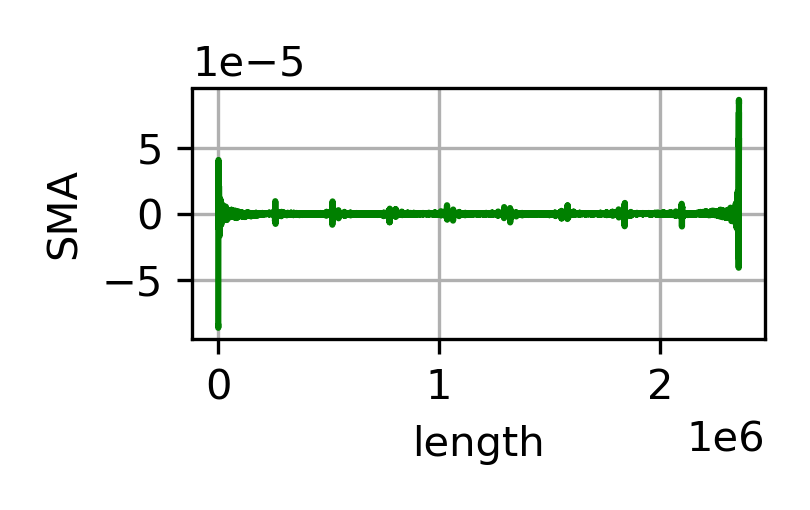}
         \caption{}
         \label{fig:sma5_18}
    \end{subfigure}
    \begin{subfigure}[b]{0.2\textwidth}
         \centering
         \includegraphics[width=\textwidth, trim=0.2cm 0.4cm 0.2cm 0]{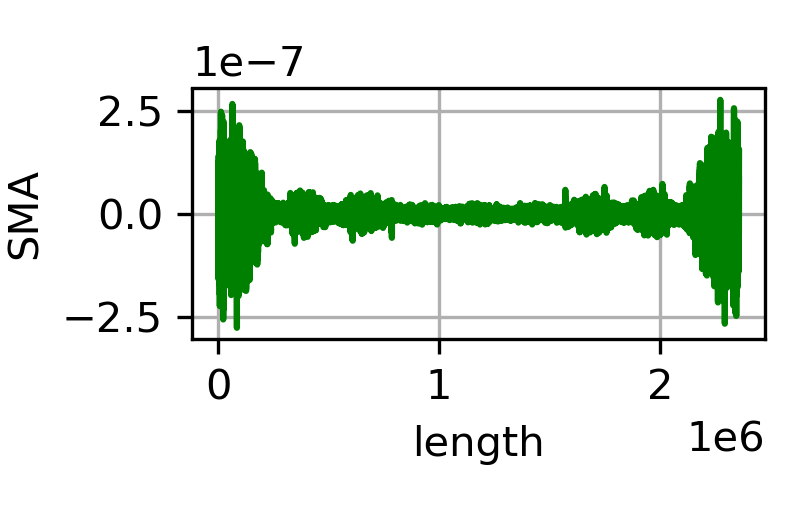}
         \caption{}
         \label{fig:sma5_19}
     \end{subfigure}
    
     \caption{SMA analysis of \textit{Convolutional} watermarked Layer belonging to the ResNet-18 (b). Subfigures (a) and (c) represent the SMA analysis of the non-watermarked layers before and after (b).}
    \label{fig:sma5}
\end{figure}

\begin{figure}[tb]
     \centering
      \begin{subfigure}[b]{0.2\textwidth}
         \centering
         \includegraphics[width=\textwidth, trim=0.2cm 0.4cm 0.2cm 0]{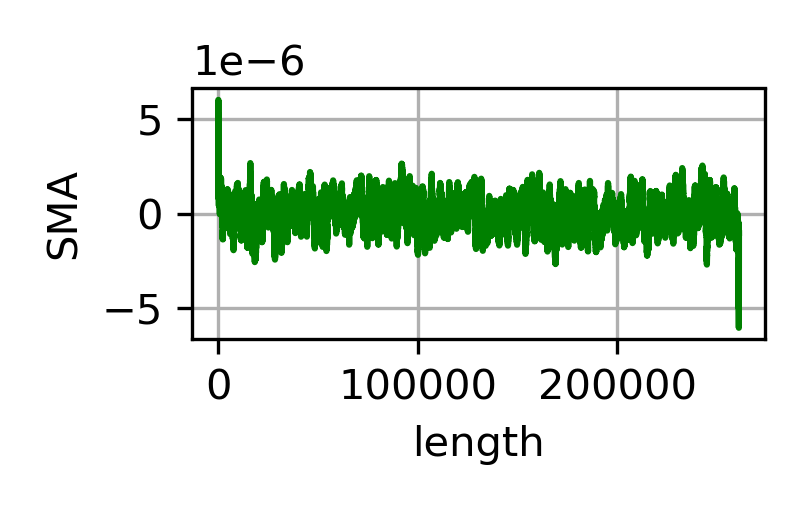}
         \caption{non-watermarked}
         \label{fig:wp_pmr_impact}
     \end{subfigure}
     \begin{subfigure}[b]{0.2\textwidth}
         \centering
         \includegraphics[width=\textwidth, trim=0.2cm 0.4cm 0.2cm 0]{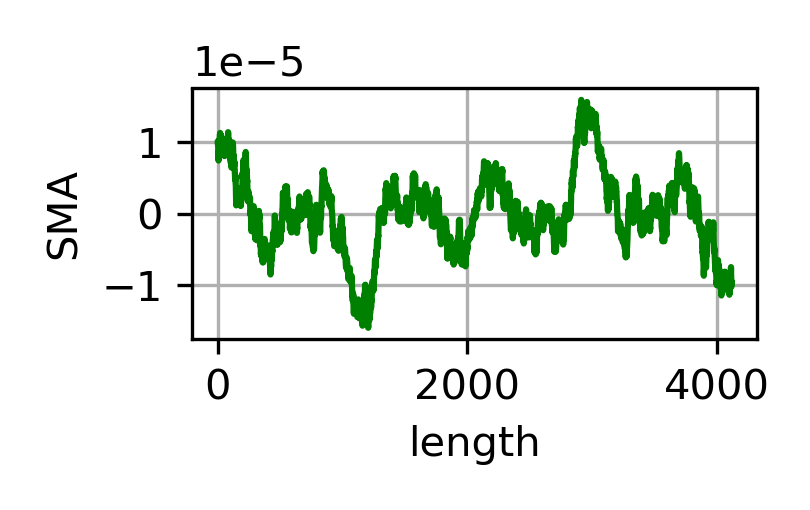}
         \caption{watermarked}
         \label{fig:watermarked}
     \end{subfigure}
        \caption{SMA analysis of \textit{Linear} watermarked and non-watermarked Layers.}
        \label{fig:linear}
\end{figure}

\section{Evaluation}
\label{sec:eval}

In this section, we demonstrate that our detection and evaluation can prevent the verification of watermarking signatures without impairing the performance of the original NN.

\subsection{Experimental Setup}

All the experiments were conducted on a server running Debian 11, with 1 TB of memory, an AMD EPYC 7742 processor with 64 physical cores and 128 threads, and 4 NVIDIA Quadro RTX 8000. We leveraged Pytorch~\cite{paszke2019pytorch} to implement the attack.
For the defenses, the environment used was suggested by their authors when available.

\noindent \textbf{Datasets. }For conducting the experiments, we used the well-known CIFAR-10~\cite{krizhevsky2009learning} and MNIST~\cite{lecun1998gradient} dataset. The CIFAR-10 dataset consists of small images ($32\times32$) of objects or animals. It includes 50k images for training and 10k for testing, depicting objects from 10 categories. The MNIST dataset is a collection of 70k grayscale images, each ($28\times28$) pixels in size, representing handwritten digits ($0-9$).

\noindent \textbf{Models. }The considered models for the evaluation include ResNet-18~\cite{he2015deep} for the CIFAR-10 dataset and a Convolutional Neural Network (CNN) for MNIST. The CNN starts with two Convolutional layers, with a 3x3 filter and padding of 1, each followed by one Max-Pooling layer of size $2\times2$ and stride of 1. Then, two Convolutional layers, with a 5x5 filter and padding of 0, each followed by one Max-Pooling layer of size $2\times2$ and stride of 1. Lastly, we added four fully connected layers with ReLU activation and another fully connected layer with Softmax activation for the classification task.
\noindent
The choice of datasets and models used for our evaluation is based on the existing and available watermarking defenses presented in the following Section.

\subsection{Evaluation Results}
\label{sub:eval_res}

In this section, we demonstrate how \ourname performs in defending against watermarking mechanisms, considering both the passive verifier (basic obfuscation) and the active verifier (advanced obfuscation). Our approach takes into account potential model modifications. We provide a detailed account of the frequency detection results, which play a pivotal role in distinguishing layers that may have been watermarked from those that have not. Then, we present the results of our approach when applied to each linear and convolutional layer under both obfuscation schemes. Finally, we delve into the discussion of how the utility of our approach may be impacted when we employ advanced obfuscation on the entire model.  

\subsubsection{Frequency Detection}
\label{subsub:freq}
This section demonstrates how the frequency components differ between watermarked and non-watermarked layer weights, as shown in Figure \ref{fig:linear} for the \textit{Linear} layers and Figure \ref{fig:sma5} for the \textit{Convolutional} layers. Each point on the plots represents the average of the model weight's frequency within a 1000-point window.
\noindent For \textit{Convolutional} layer analysis, we watermarked different layers of ResNet-18 and analyzed the average (SMA) frequency patterns of the \textit{Convolutional} layers.

\noindent
For \textit{Linear} layer analysis, we tested watermarking on one or more \textit{Linear} layers in the CNN. We concluded that, for the \textit{Linear} (Figure \ref{fig:linear}), the watermarked layers exhibit a higher degree of stability in the average change of frequency values, with fewer fluctuations compared to the weights of the non-watermarked layers. Similarly, we observed less volatility for the SMA values for the \textit{Convolutional} layers, except for the sudden high jumps in values at the extreme points. Thus, we were able to discriminate the watermarked layers from the non-watermarked layers.

\begin{table*}[t]
\centering
\scalebox{0.72}{
\begin{tabular}{c|cc|cc|cc|cc|cc|cc|cc|cc|}
          & \multicolumn{2}{c|}{DeepMarks~\cite{chen2019deepmarks}} & \multicolumn{2}{c|}{DeepSigns~\cite{darvish2019deepsigns}} & \multicolumn{2}{c|}{Feng \etal~\cite{feng2020watermarking}} & \multicolumn{2}{c|}{Kuribayashi \etal~\cite{kuribayashi2021white}} & \multicolumn{2}{c|}{DeepMark~\cite{xie2021deepmark}} & \multicolumn{2}{c|}{NeuNAC~\cite{botta2021neunac}} & \multicolumn{2}{c|}{Sablayrolles \etal~\cite{sablayrolles2020radioactive}} & \multicolumn{2}{c|}{DeepAttest~\cite{chen2019deepattest}} \\
          & Utility      & Conf.     & Utility      & Conf.     & Utility       & Conf.      & Utility          & Conf.          & Utility     & Conf.     & Utility    & Conf.    & Utility           & Conf.          & Utility      & Conf.      \\ \hline
No Obf.   & 99.7         & 100            & 96.9         & 100            & 91.4          & 98.4            & 98.3             & 98.5                & 99.1        & 100            & 99.6       & 100           & 94.7              & 99.3                & 99.5         & 100             \\
Base Obf  & 99.7         & N.S.
& 96.9         & N.S.
& 91.4          & 42.1            & 98.3             & 19.1                & 99.1        & N.S.
& 99.6       & 51.2          & 94.7              & N.S.                & 99.5         & N.S.
\\
Adv. Obf. & 98.8         & 37.5
& 96.4         & 43.8
& 88.7          & 44.6            & 98.1             & 23.5                & 98.4        & 41.3           & 97.8       & 48.1          & 93.6              & 45.0
& 98.7         & 37.8           
\end{tabular}}
\caption{Evaluation results of the watermarking schemes, based on the MNIST dataset, when the watermark is embedded in the \textit{Linear} Layer. Confidence (Conf.) represents the probability that the extracted signature matches the one embedded by the real owner. The Utility of the CNN before watermarking was 99.8\%. All the values in percentage (\%).}
\label{tab:mnist_evaluation}
\end{table*}
\begin{table*}[t]
\centering
\scalebox{0.63}{
\begin{tabular}{c|cc|cc|cc|cc|cc|cc|cc|cc|cc|cc|}
          & \multicolumn{2}{c|}{Uchida \etal~\cite{Uchida_2017}} & \multicolumn{2}{c|}{DeepMarks~\cite{chen2019deepmarks}} & \multicolumn{2}{c|}{Feng \etal~\cite{feng2020watermarking}} & \multicolumn{2}{c|}{Guan \etal~\cite{guan2020reversible}} & \multicolumn{2}{c|}{DeepMark~\cite{xie2021deepmark}} & \multicolumn{2}{c|}{NeuNAC~\cite{botta2021neunac}} & \multicolumn{2}{c|}{Liu \etal~\cite{liu2021watermarking}} & \multicolumn{2}{c|}{Lottery~\cite{chen2021you}} & \multicolumn{2}{c|}{RIGA~\cite{wang2021riga}} & \multicolumn{2}{c|}{DeepAttest~\cite{chen2019deepattest}} \\
          & Utility       & Conf.       & Utility         & Conf.         & Utility       & Conf.      & Utility      & Conf.      & Utility        & Conf.        & Utility       & Conf.       & Utility          & Conf.         & Utility        & Conf.        & Utility       & Conf.      & Utility         & Conf.         \\ \hline
No Obf.   & 91.9          & 99.9        & 91.8            & 100          & 91.4          & 99.6       & 91.1         & 89.7       & 87.7           & 100          & 93.6          & 95.4        & 93.3             & 100           & 93.7           & 100          & 89.4          & 100        & 92.3            & 100           \\
Base Obf  & 91.9          & N.S.        & 91.8            & N.S.         & 91.4          & 43.1       & 91.1         & 36.6       & 87.7           & 42.3         & 93.6          & N.S
& 93.3             & 49.6          & 93.7           & N.S.
& 89.4          & N.S.       & 92.3            & N.S.          \\
Adv. Obf. & 89.8          & 41.3        & 88.7            & 34.2         & 90.9          & 35.7       & 87.4         & 24.9       & 85.6           & 40.7         & 91.4          & 43.0          & 93.1             & 47.9          & 92.8           & 31.3         & 86.7          & N.S.       & 90.7            & N.S.         
\end{tabular}
}
\caption{Evaluation results of the watermarking scheme, based on the CIFAR-10 dataset, with watermarking on the \textit{Convolutional} Layer. Confidence (Conf.) represent the probability that the extracted signature match the one embedded by the real owner. The Utility of the ResNet-18 before watermarking was 94.2\%. All the values in percentage (\%).}
\label{tab:cifar_eval}
\end{table*}
\subsubsection{\ourname Attack on Linear Layers}
\label{subsub:linear}
To evaluate our approach against watermark defenses that embed their signature in the \textit{Linear} layers, we will consider a CNN trained on MNIST~\cite{lecun1998gradient}. Following Table~\ref{tab:mnist_evaluation}, we outline the drop in watermarking verification's accuracy using \ourname's base and advanced obfuscations, or if the signature extraction is completely prevented, e.g., due to incompatible matrix multiplication, the verifier cannot proceed (marked as no signature, $N.S.$). To simplify the evaluation process, we categorize the evaluated watermark frameworks on how their approaches extract the secret signature during the passive verification and how they attempt to bypass the obfuscation during active verification.

\noindent \textbf{Weight-Based Watermarks. }This category of watermarks~\cite{chen2019deepmarks, botta2021neunac, chen2019deepattest} embeds the secret message into the weights of $F_i$ (Section \ref{sec:background} and \ref{subsec:detail_design}). During passive verification, they use a secret matrix $S \in \mathbb{R}^{n \times n}$ (or a secret vector $\overrightarrow{w} \in \mathbb{R}^{n}$), kept by the model owner, to obtain the secret message by multiplying the selected layer weights $A_{i}$ (or statistical features of $A_{i}$) with this secret: $A_{i} \times S$ or $A_{i} \times \overrightarrow{w}$. However, as shown in Figure~\ref{fig:base_linear}, \ourname transforms $F_i$ into $F^{1}_i$. Therefore, when the passive verifier tries to multiply ($A_{i} (m \times n) \times H (n \times h)$) by $S$ (or $\overrightarrow{w}$), the dimensions will be incompatible ($h \neq n$), failing the extraction. In the active verifier setting, as illustrated in Figure~\ref{fig:adv_linear} and in Section \ref{subsec:detail_design}, we assume that the active verifier tries to reshape the weights of $F^{1}_i$ or $F^{2}_i \times F_{i+1}$ to match the dimensions of $S$ (or $\overrightarrow{w}$). While it is possible, however, the content of both layers has no relation to the original $A_{i}$ (as well as different statistical features). Therefore, as shown in Table~\ref{tab:mnist_evaluation}, the resulting watermarking confidence will be akin to a random guess.

\noindent \textbf{Weight-Selection Watermarks. }This category of watermarks~\cite{feng2020watermarking, kuribayashi2021white} chooses values at specific positions in the layer $F_i$. Thus, during verification, the passive verifier searches for those values in the same positions as in $A_{i}$ and calculates the secret message from the vector composed of those values. Since \ourname transforms the layer $F_i$ into $F^{1}_i$ with dimension $\mathbb{R}^{m \times h}$ and also updated its weight matrix with random values, which has no direct relation to the original $A_{i}$; therefore, when the passive verifier tries to locate the watermarked elements, it will extract random weights, obtaining low watermarking confidence that disproves the model ownership. When we assume an active verifier who try to multiply $F^{1}_i$ with $F^{2}_{i}$ to obtain the original $F_i$, we show that \ourname advanced obfuscation can counteract this reversing by replacing $F^{2}_{i}$ with $F^{2}_{i} \times F_{i+1}$, as detailed in Section \ref{subsec:detail_design}. When the active verifier extracts the signature weights again, they will observe that it has no relation to the original $A_{i}$ (as well as different statistical features). Therefore, as shown in Table~\ref{tab:mnist_evaluation}, the resulting watermarking confidence drops to a random guess.

\noindent \textbf{Activation-Based Watermarks. }This category of watermarks~\cite{darvish2019deepsigns, sablayrolles2020radioactive, xie2021deepmark} evaluates the activation maps of selected layer $F_i$ for a hidden dataset of selected inputs. To extract the signature from the output of $F_i$, the passive verifier generates a secret transformation matrix $S \in \mathbb{R}^{n \times n}$ that is dimensionally compatible with $A_{i}$ ($\mathbb{R}^{m \times n}$). But as shown in Figure~\ref{fig:base_linear}, \ourname transforms $F_i$ into $F^{1}_i$ into a new layer with weights $A_i \times H \in \mathbb{R}^{m \times h}$, leading the secret matrix $S$ to be dimensionally incompatible with $F^{1}_i$ ($h \neq n$). Consequently, the passive verifier fails to extract the signature from the model. In the advanced scenario, we assume that the active verifier tries to reverse back $F_i$ from $F^{1}_i$ and $F^{2}_i$, but \ourname's advanced obfuscation preventively merges $F^{2}_i$ with $F_{i+1}$ into $F^{2}_i\times F_{i+1}$. Thus, even if the active verifier tries to reshape the outputs to match $S$, the activation maps of both layers will still contain enough noise, lowering the watermark confidence, as shown in Table~\ref{tab:mnist_evaluation}.

\subsubsection{\ourname Attack on Convolutional Layers}
\label{subsub:convolutional}
To evaluate our approach against watermark defenses that embed their signature in the \textit{Convolutional} layers, we consider a ResNet-18~\cite{he2015deep} trained on CIFAR-10~\cite{krizhevsky2009learning}. In Table~\ref{tab:cifar_eval}, we outline the drop in watermarking verification's accuracy with the use of \ourname's base and advanced obfuscations, or if the signature extraction is completely prevented ($N.S.$). Again, we categorize the evaluated watermarks based on how their approaches extract the secret signature from the \textit{Convolutional} layers during the passive verification and how an active verifier would attempt to bypass the obfuscation during verification.

\noindent \textbf{Weight-Based Watermarks. }This category of watermarks~\cite{Uchida_2017, wang2021riga, chen2019deepmarks, botta2021neunac, chen2021you, chen2019deepattest} embeds the secret message into the weights of one or more \textit{Convolutional} layers. At verification time, the passive verifier extracts a vector $\overrightarrow{w} \in \mathbb{R}^{C\cdot k_h\cdot k_w}$ that is the flattened average of one or more $Conv_i \in \mathbb{R}^{P \times C\times k_h\times k_w}$ for each feature maps $p \in P$. Then, a signature is obtained by computing a similarity (Equation \ref{eq:verification}) between $\overrightarrow{w}$ (in some cases using trained DNN~\cite{wang2021riga}), and the owner's secret signature $M_o$. Since we expand the kernel, as illustrated in Figure~\ref{fig:base_conv}, \ourname forces the extraction of a signature of incorrect size $C\cdot k^{'}_{h}\cdot k^{'}_{w}$ that cannot be compared with the owner's signature, halting the verification. An active verifier could detect the padding in the kernel, but by carefully adding noise, as explained in Equation \ref{eq:conv_varepsilon}, \ourname prevents the verifier from determining the position of the original kernel to extract. The only option left for an active verifier, in order to obtain a functional $M_e$, is to reshape the kernels to match them with the owner's signature size. However, this operation does not help the active verifier, given that, as shown in Table~ \ref{tab:cifar_eval}, the computed watermarking confidence with the reshaped kernels is still too low to justify a claim of ownership.

\noindent \textbf{Activation-Based Watermark. }This category of watermarks~\cite{guan2020reversible, liu2021watermarking, xie2021deepmark} evaluates the resulting feature maps for a hidden dataset of selected inputs given to the selected \textit{Convolutional} layers $Conv_i$. To extract the signature, the passive verifier generates a secret matrix $S$ that is dimensionally compatible with the feature map of the \textit{Convolutional} layer. During the verification, \ourname base obfuscation changes the statistical values that the defense is expected to obtain, reaching watermark confidence that cannot support the claim of ownership. In the active verifier scenario, the original kernels remain hidden during the inspection, producing again a lower watermark confidence than the base obfuscation case, as shown in Table~\ref{tab:mnist_evaluation}. 

\begin{table}[!t]
\centering
\scalebox{0.60}{
\begin{tabular}{clccc}
\multicolumn{1}{c}{\textbf{Category}} & \textbf{Attack}  & \textbf{Data}             & \textbf{Retraining}           & \textbf{Model buildout} \\ \hline
\multirow{14}{*}{\textbf{\begin{tabular}[c]{@{}c@{}}Model \\ Modification\end{tabular}}}  & Adversarial Training~\cite{madry2019deep} &\checkmark & \checkmark &  $\times$ \\
                                      & Fine-Tuning~\cite{Uchida_2017} & \checkmark & \checkmark & $\times$ \\
                                      & Feature Permutation~\cite{lukas2022sok} &  \checkmark & \checkmark & $\times$ \\
                                      & Fine Pruning~\cite{liu2018finepruning} &\checkmark & \checkmark & $\times$ \\
                                      & Overwriting~\cite{Uchida_2017} & \checkmark & \checkmark & $\times$ \\
                                      & Regularization~\cite{regularization_shaf} & \checkmark & \checkmark & $\times$ \\
                                      & Label Smoothing~\cite{Rethinking_LabelSmoothing} & \checkmark & \checkmark & $\times$ \\
                                      & Neural Structure Obfuscation~\cite{yan2023rethinking} & $\times$ & $\times$ & \checkmark \\ 
                                      & Neural Cleanse~\cite{wang2019neural} & \checkmark & \checkmark &  \checkmark \\
                                      & Neural Laundering~\cite{aiken2020neural} & \checkmark & \checkmark & \checkmark \\
                                      & Weight Pruning~\cite{zhu2017prune} & \checkmark & \checkmark & $\times$ \\
                                      & Weight Shifting~\cite{lukas2022sok} & \checkmark & \checkmark & $\times$ \\
                                      & Weight Quantization~\cite{hubara2016quantized} & \checkmark & \checkmark & $\times$ \\
                                      & \textbf{\ourname} & $\times$ & $\times$ & $\times$ \\\hline
\multicolumn{1}{c}{\textbf{\begin{tabular}[c]{@{}c@{}}Model \\ Extraction\end{tabular}}}   & Distillation~\cite{hinton2015distilling} & \checkmark & \checkmark & \checkmark \\ 
\end{tabular}
}

\caption{"Data" represents attacks that require training data, "Retraining" for Training and Fine-Tuning, and "Model Buildout" if neurons are modified, a new model is created.} 

\label{tb:rw_whitebox}
\end{table}

\subsubsection{Runtime Evaluation}
\label{subsub:runtime}
To evaluate the applicability of \ourname in real-case scenarios, we evaluated the increase in complexity and runtime execution of obfuscated models with regard to the original stolen models. Table~\ref{tab:runtime} shows that our approach does not significantly increase the execution time after the injection of the watermarking. Instead, for the advanced obfuscation, the multiplication of layers by a smaller $H$ reduces the execution time.

\begin{table}
\centering
\scalebox{0.80}{
\begin{tabular}{l|c|c|}
\multicolumn{1}{c|}{Runtime Execution}              & ResNet-18 & CNN \\ \hline
Non-Watermarked Model           & 1.09 s         & 0.48 s \\
Watermarked Model           & 1.24 s        & 0.56 s \\
Base Obfuscation        & 1.31 s       & 0.68 s  \\
Advanced Obfuscation        & 1.21 s        & 0.51 s  
\end{tabular}
}
\caption{Average execution time across multiple runs, using different seeds over the test partition of the respective datasets.}
\label{tab:runtime}
\end{table}
\section{Related Works}
\label{sec:related_works}

This section provides an overview of the removal attacks on white-box watermarking schemes. The attacks can be subdivided into the following categories (Lukas N. \etal~\cite{lukas2022sok}): (i) input preprocessing, (ii) model modification, and (iii) model extraction. We will present only watermarking schemes that check the model composition.

\noindent \textbf{Model Modification. }In this category, the adversary first transforms a model to introduce vulnerabilities that can be further exploited using various techniques such as fine-tuning, pruning and more~\cite{madry2019deep, shafieinejad2021robustness, Rethinking_LabelSmoothing, Uchida_2017, liu2018finepruning, zhu2017prune, lukas2022sok, hubara2016quantized, wang2019neural, aiken2020neural, yan2023rethinking}.

\noindent
In~\cite{madry2019deep}, they first select a subset of the dataset, which they modify using Projected Gradient Descent (PGD), causing disruption to the watermark. Subsequently, they fine-tuned the model, utilizing the true labels for the adversarial examples, to ensure that the model's predictions aligned more closely with the underlying data. 

\noindent
In the case of the regularization attack~\cite{shafieinejad2021robustness}, the authors initially applied a regularization technique, such as L2, to the model. They then fine-tuned the model to discover a new optimal local minimum, aiming to reduce the drop in accuracy caused by the initial regularization step. Due to the modifications made to the model's weights, the watermark has been removed. 
Szegedy \etal~\cite{Rethinking_LabelSmoothing} introduced Label Smoothing, a method where uncertainty was introduced during training. This was achieved by combining one-hot encoded or predicted labels with a uniform distribution over all possible labels. This approach served to regularize the model, treating the ground-truth class label and other potential labels equally likely. Because the regularization process influenced the model's layers, the resultant modifications rendered the watermark less detectable.

\noindent
In~\cite{Uchida_2017}, the authors employed fine-tuning to modify the model's weights and introduced various fine-tuning approaches aimed at removing the watermark from the model. These approaches include (i) fine-tuning all layers, (ii) fine-tuning only the last layer after freezing the preceding layers, (iii) retraining all layers by re-initializing the weights of the last layers and then fine-tuning all weights, and (iv) retraining only the last layer. Through the fine-tuning process, the layers containing the watermark are updated, resulting in a failure during the watermark verification. They also presented another approach, which embeds a watermark using the same watermarking scheme but a different watermarking key, thus overwriting the previous one. 

\noindent
In~\cite{liu2018finepruning}, the authors introduced the fine-pruning method for removing the watermark, which involves a combination of fine-tuning and pruning to reduce the effectiveness of the watermarking, thus eliminating it. Initially, the pruning step entailed setting the activation of neurons with low activation for benign samples to zero. Then, the model was fine-tuned to mitigate the drop in test accuracy. Considering the impact of these two approaches independently on watermark removal, their combination also led to the failure of the watermark. 

\noindent
Another set of methods focuses on altering model weights to eliminate or obscure watermarked layers. For instance, Zhu M. \etal~\cite{zhu2017prune} employ the weight pruning technique, which involves randomly pruning weights in the model until a specific sparsity level is achieved within each layer. In another study~\cite{lukas2022sok}, authors introduced the weight shifting technique, where they apply minor perturbations to the filters of convolutional networks. Subsequently, the model is fine-tuned to enhance test accuracy.

\noindent
In the Weight Quantization approach~\cite{hubara2016quantized}, the authors utilize $b$ bits to quantize the model's weights. These weights are then divided into $2^b$ equally spaced segments, with each weight parameter approximated by the central value of its respective segment. All of these aforementioned approaches, ~\cite{zhu2017prune, lukas2022sok, hubara2016quantized} operate on the assumption that, in typical scenarios, watermarks are embedded within the model weights or architecture. Therefore, modifying these components weakens the watermarking. 
Moreover, Lukas N. \etal~\cite{lukas2022sok} presented the feature permutation technique, where neurons in a hidden layer are randomly rearranged without affecting the model's functionality. This is possible due to the feature invariance of deep neural networks. Altering the positions of the weights within the neural network also has an impact on the watermark, as it directly affects the layer(s) containing it. 

\noindent
Neural Cleanse~\cite{wang2019neural} and Neural Laundering~\cite{aiken2020neural} are approaches based on reverse engineering. Firstly, the watermarking is removed by reverse-engineering the watermarking trigger from the model, eliminating the trigger. After this process, the trigger can be unlearned by fine-tuning it on different labels or by iteratively pruning the most active neurons in some layers. The distinction between Neural Cleanse and Neural Laundering is that the first was developed for removing backdoors, while the latter used the same approach to remove the watermark. 

\noindent
Yifan Y. \etal~\cite{yan2023rethinking}, presented a Neural Structural Obfuscation approach, which uses dummy neurons to perform neural structural obfuscation, automatically generating and injecting dummy neurons inside the target model to reduce the success of watermark verification. While this approach works against the passive verifier, this approach will not work against an active verifier as described in Section~\ref{sec:threat}. This active verifier could easily detect the added dummy neuron by comparing the weights of the original model and the obfuscated one, as described in Section~\ref{sec:design}.

\noindent \textbf{Model Extraction. }In this watermark removal scheme, the knowledge from the source model is transferred into a surrogate model. However, this category lies in the black-box domain except for the Distillation\cite{hinton2015distilling}, which introduces model compression by transferring knowledge from a larger teacher neural network to a smaller student network.

\noindent All the approaches discussed above have certain limitations, primarily based on their demands for data and resources to facilitate training, fine-tuning, or modification to remove the watermark embedded in the image or the model. For instance, in the case of Model Modification, the model needs to be modified to perform the necessary modification, leading to a need for input data and resources to perform the fine-tuning to compensate for the drop in accuracy. Lastly, Model Extraction employs a distillation process wherein the teacher model must be reduced to a student model during the training. This, again, requires the availability of suitable data and computational resources. Table \ref{tb:rw_whitebox} outlines these reported watermark removal attacks, detailing the limitations regarding the prior information they need (data and fine-tuning or retraining) to function effectively.
\section{Security Considerations}
\label{sec:sec_analysis}

\noindent
This section considers the effectiveness of our obfuscation techniques, corroborating that \ourname can neutralize both the passive and active verifiers without requiring prior knowledge of the underlying watermarking schemes, additional data or training, and fine-tuning. To bypass our attack scheme, an active verifier has to ensure that they can either undo the obfuscation schemes or compute a different set of operations that can help them identify the modifications done on the stolen model. Thus succeeding in verifying the watermark. Below, we present the schemes through which a verifier can determine that the given model is indeed a stolen copy of the original model: i) Keep track of the original layer's weight values and their positions in the $\mathbb{R}^{m \times n}$ matrix, ii) Train and deploy an extra deep neural network (DNN) to extract the hidden message from the layer's weights~\cite{wang2021riga}, iii) Analyze the activation maps of the layer for specific input instances~\cite{darvish2019deepsigns, guan2020reversible}, or iv) Perform an extra set of computations to compare the original and the assumed stolen model. Thus, an advanced verifier with knowledge of the stolen model architecture attempts to compute $H$ from $A_i \times H$ using $A_i$ and then obtain the identity matrix from the subsequent $H^{-1}$.

\noindent Thus, the adversary, when attacking the watermarking approaches, can alter the weight matrix or alter the layer's output. \ourname encompasses these attack methods through two obfuscation schemes presented in Section \ref{sec:design}. These schemes reshape the entire watermarked area of the model by splitting and expanding using $I_{n \times n}$ for Base Linear Layer Obfuscation, merging $F^{2}{i}$ with layer $F^{1}{i+1}$ for Advanced Linear Layer Obfuscation, modifying the kernel shape ($k_h \times k_w$) and size using padding for Base Convolutional Layer Obfuscation, and introducing noise $\lambda,\; \varepsilon$ for further alteration in the Advanced Convolutional Layer Obfuscation. Against an advanced verifier that tries to obtain $H$ from $A_i \times H$, the adversary can expand the model architecture by incorporating new layers between $A_i \times H$ and $H^{-1}$ using the obfuscation schemes presented in Section \ref{sec:design}, and, if the adversary does not want to introduce noise, by cleverly introducing another $H^{-1}$ in the middle of the expansion it can deceive the verifier into thinking it has obtained the identity matrix before the expected position.

\noindent
We also empirically verified the effectiveness of \ourname in Section \ref{sec:eval}. The results reveal that \ourname effectively disrupts multiple such schemes, reducing watermark detection to the level of random guessing while maintaining model accuracy. Therefore, \ourname is robust against both passive and active verifiers. 
\section{Conclusion}
\label{sec:conc}

This paper presents \ourname, a unified white-box watermark obfuscation framework. Unlike existing methods, it offers both basic and advanced obfuscation techniques, catering to different scenarios. The basic method targets passive verifiers adhering to standard protocols, while the advanced method addresses situations where active verifiers have access to the entire model and can perform additional computations. Extensive evaluations were conducted, testing \ourname against various white-box watermarking methods. The results reveal that \ourname effectively disrupts multiple schemes discussed in the paper, reducing watermark detection to the level of random guessing while maintaining model accuracy.

\section*{Acknowledgment}

\noindent This research received funding from Intel through the Private AI Collaborate Research Institute (https://www.private-ai.org/), the Hessian Ministry of Interior and Sport as part of the F-LION project, OpenS3 Lab, and the Horizon program of the European Union under the grant agreement No. 101093126 (ACES).

\bibliographystyle{plain}
\bibliography{reference}

\appendix

\section{Algorithmic Implementation}
Algorithmic implementation of \ourname: starting from Algorithm~\ref{alg:base}, if the SMA detection algorithm has not found any potential layer, we apply the obfuscation to the whole model (line 2); otherwise, for each watermarked detected layer, we apply the specific obfuscation based on its type and accordingly to the kind of verification that is expected. The function AdvConvObfuscation (Algorithm~\ref{alg:advconv}) calculates an $\epsilon$ value for each feature of the layer's feature map and uses it for padding; it then multiplies the whole layer by a random value and divides the subsequent one by the same value.
The function BaseConvObfuscation (Algorithm~\ref{alg:baseconv}) pads each feature map of the kernel with zeros.
The function AdvLinearObfuscation (Algorithm~\ref{alg:advlinear}) calculates the matrices $H$ and $H^{-1}$ based on the dimension of the watermarked layer and the subsequent one; it then multiplies the original layer's weight by $H$ and the subsequent layer by $H^{-1}$. Lastly, function BaseLinearObfuscation (Algorithm~\ref{alg:baselinear}) calculates the matrices $H$ and $H^{-1}$ based on the dimension of the watermarked layer; it then expands the model on that specific layer while keeping the previous and subsequent layers untouched.

\begin{algorithm}
\caption{AdvLinearObfuscation($M$,  $layer$)\\
\textbf{Input:} \\
$\hphantom{xxxx}M$ Model to obfuscate\\
$\hphantom{xxxx}layer$ position of the Linear layer\\
\textbf{Output:} \\
$\hphantom{xxxx}M$ Model with advanced obfuscation of Linear layer}
\label{alg:advlinear}
\begin{algorithmic}[1]

\STATE $H, H^{-1} \leftarrow RandomMatrix(M[layer], M[layer+1])$
\STATE $M[layer] \leftarrow M[layer] \times H$
\STATE $M[layer+1] \leftarrow M[layer+1] \times H^{-1}$

\RETURN $M$
\end{algorithmic}
\end{algorithm}

\begin{algorithm}
\caption{BaseLinearObfuscation($M$,  $layer$)\\
\textbf{Input:} \\
$\hphantom{xxxx}M$ Model to obfuscate\\
$\hphantom{xxxx}layer$ position of the Linear layer\\
\textbf{Output:} \\
$\hphantom{xxxx}M'$ Model with base obfuscation of Linear layer}
\label{alg:baselinear}
\begin{algorithmic}[1]

\STATE $H, H^{-1} \leftarrow RandomMatrix(M[layer])$

\STATE $M', newLayer \leftarrow expandModel(M, layer) $
\STATE $M'[layer] \leftarrow M'[layer] \times H$
\STATE $M'[newLayer] \leftarrow H^{-1}$

\RETURN $M'$
\end{algorithmic}
\end{algorithm}

\begin{algorithm}
\caption{AdvConvObfuscation($M$,  $layer$)\\
\textbf{Input:} \\
$\hphantom{xxxx}M$ Model to obfuscate\\
$\hphantom{xxxx}layer$ position of the Conv. layer\\
\textbf{Output:} \\
$\hphantom{xxxx}M$ Model with advanced obfuscation of Conv. layer}
\label{alg:advconv}
\begin{algorithmic}[1]

\FOR{$p$ in $FeatureMap(M[layer])$}     
    \STATE $\epsilon \leftarrow \beta * min(|M[layer][p]|) * \mathcal{N}(\mu,\,\sigma^{2})$
    \STATE $M[layer][p] \leftarrow PadKernel(M[layer][p], \epsilon)$
\ENDFOR
\STATE $\lambda \leftarrow random()$
\STATE $M[layer] \leftarrow M[layer] * \lambda$
\STATE $M[layer + 1] \leftarrow M[layer + 1] / \lambda$

\RETURN $M$
\end{algorithmic}
\end{algorithm}

\begin{algorithm}
\caption{BaseConvObfuscation($M$,  $layer$)\\
\textbf{Input:} \\
$\hphantom{xxxx}M$ Model to obfuscate\\
$\hphantom{xxxx}layer$ position of the Conv. layer\\
\textbf{Output:} \\
$\hphantom{xxxx}M$ Model with base obfuscation of Conv. layer}
\label{alg:baseconv}
\begin{algorithmic}[1]

\FOR{$p$ in $FeatureMap(M[layer])$}     
    \STATE $M[layer][p] \leftarrow PadKernel(M[layer][p], 0)$
\ENDFOR

\RETURN $M$
\end{algorithmic}
\end{algorithm}

\begin{algorithm}
\caption{\ourname($M$,  $SMA$, $Adv$)\\
\textbf{Input:} \\
$\hphantom{xxxx}M$ original stolen Model\\
$\hphantom{xxxx}SMA$ detection's output (can be empty)\\
$\hphantom{xxxx}Adv$ if we apply advanced obfuscation\\
\textbf{Output:} \\
$\hphantom{xxxx}M'$ obfuscated Model}
\label{alg:base}
\begin{algorithmic}[1]
\IF{$SMA = \emptyset$}
    \STATE $SMA \leftarrow \{layer\;|\;layer \in M\}$
\ENDIF

\STATE $M' \leftarrow copy(M)$

\FOR{$layer$ in $SMA$}
    \IF{$layer$ is Convolutional}        
        \IF{$Adv$ is True}
            \STATE $M'[layer] \leftarrow AdvConvObfuscation(M', layer)$
        \ELSE
            \STATE $M'[layer] \leftarrow BaseConvObfuscation(M', layer)$
        \ENDIF
    \ELSIF{$layer$ is Linear}        
        \IF{$Adv$ is True}
            \STATE $M'[layer] \leftarrow AdvLinearObfuscation(M', layer)$
        \ELSE
            \STATE $M'[layer] \leftarrow BaseLinearObfuscation(M', layer)$
        \ENDIF
    \ENDIF
\ENDFOR

\RETURN $M'$
\end{algorithmic}
\end{algorithm}

\end{document}